\documentclass[%
 reprint,
nofootinbib,
 amsmath,amssymb,
 aps,
floatfix,
twocolumn,
]{revtex4-1}

\usepackage[utf8]{inputenc}
\usepackage[euler-digits,euler-hat-accent]{eulervm}
\usepackage{newpxtext}

\usepackage{amsthm, thmtools}
\usepackage{xfrac}
\usepackage{algorithm}
\usepackage{listings}
\usepackage[colorlinks=false]{hyperref}
\usepackage{nameref, cleveref}
\crefname{ineq}{inequality}{inequalities}
\Crefname{ineq}{Inequality}{Inequalities}
\crefname{figure}{Fig.}{Figs.}
\usepackage{graphicx}
\usepackage{subcaption}

\usepackage{todonotes}
\usepackage{mathtools}
\DeclarePairedDelimiter\bra{\langle}{\rvert}
\DeclarePairedDelimiter\ket{\lvert}{\rangle}
\DeclarePairedDelimiterX\braket[2]{\langle}{\rangle}{#1 \delimsize\vert #2}
\DeclarePairedDelimiter\set{\{}{\}}
\DeclarePairedDelimiter\ceil{\lceil}{\rceil}
\DeclarePairedDelimiter\floor{\lfloor}{\rfloor}
\DeclarePairedDelimiter\paren{(}{)}

\usepackage{algpseudocode}
\usepackage{tikz}
\usetikzlibrary{quantikz}

\declaretheorem[name=Theorem, refname={Theorem,Theorems}, Refname={Theorem, Theorems}]{thm}
\declaretheorem[name=Lemma, refname={Lemma, Lemmas}, Refname={Lemma, Lemmas}, sibling=thm]{lem}
\declaretheorem[name=Proposition, refname={Proposition, Propositions}, Refname={Proposition, Propositions}, sibling=lem]{prop}
\declaretheorem[name=Definition, refname={Definition, Definitions}, Refname={Definition, Definitions}, style=definition, sibling=prop]{ddef}
\declaretheorem[name=Example, refname={Example, Examples}, Refname={Example, Examples}, style=remark, sibling=ddef]{exm}
\declaretheorem[name=Remark, refname={Remark, Remarks}, Refname={Remark, Remarks}, style=remark, sibling=exm]{remark}
\declaretheorem[name=Observation, refname={Observation, Observations}, Refname={Observation, Observations}, sibling=lem]{observation}
\declaretheorem[name=Corollary, refname={Corollary, Corollaries}, Refname={Corollary, Corollaries}, style=remark, sibling=remark]{corollary}

\newcommand{\PP}{\mathbb{P}}
\newcommand{\calO}{\mathcal{O}}
\newcommand{\calP}{\mathcal{P}}
\newcommand{\calA}{\mathcal{A}}
\DeclareMathOperator{\Ima}{Im}
\DeclareMathOperator{\rank}{rank}
\DeclareMathOperator{\target}{target}
\DeclareMathOperator{\id}{Id}
\DeclareMathOperator{\Swap}{SWAP}
\DeclareMathOperator{\cnot}{CNOT}

\hypersetup{
    pdftitle={Introducing structure to expedite quantum search},
    pdfauthor={Marcin Briański, Jan Gwinner, Vladyslav Hlembotskyi, Witold Jarnicki, Szymon Pliś, Adam Szady}
}

\begin{document}

\title{Introducing structure to expedite quantum search}
\author{Marcin~Briański}
\email{marbri@beit.tech}
\author{Jan~Gwinner}
\email{jan.gwinner@beit.tech}
\author{Vladyslav~Hlembotskyi}
\email{vlad@beit.tech}
\author{Witold~Jarnicki}
\email{witek@beit.tech}
\author{Szymon~Pliś}
\email{szymon@beit.tech}
\author{Adam~Szady}
\email{adsz@beit.tech}
\affiliation{Beit.tech}

\date{\today}

\begin{abstract}
    We present a novel quantum algorithm for solving the unstructured search problem with one marked element. Our algorithm allows generating quantum circuits that use asymptotically fewer additional quantum gates than the famous Grover's algorithm and may be successfully executed on NISQ devices. We prove that our algorithm is optimal in the total number of elementary gates up to a multiplicative constant. As many NP-hard problems are, in fact, not unstructured, we also describe the \emph{partial uncompute} technique which exploits the oracle structure and allows a significant reduction in the number of elementary gates required to find the solution. Combining these results allows us to use an asymptotically smaller number of elementary gates than Grover's algorithm in various applications, keeping the number of queries to the oracle essentially the same. We show how the results can be applied to solve hard combinatorial problems, for example, Unique \(k\)-SAT. Additionally, we show how to asymptotically reduce the number of elementary gates required to solve the unstructured search problem with multiple marked elements.
\end{abstract}

    \maketitle

\section{Introduction}
    In the quantum \emph{unstructured search problem} the task is to find one marked element out of N elements corresponding to the computational basis. We want to accomplish that by the least possible number of queries to a given phase oracle, the only action of which is changing the signs of the coordinates corresponding to the marked elements. For more details, see \cref{section:premilinaries}.

    The celebrated Grover's algorithm \cite{grover96} is one of the main achievements of quantum computing. It locates a marked element using only $\calO(\sqrt{N})$ queries to the oracle and $\calO(\sqrt{N} \log{N})$ additional (i.e. non-oracle) elementary gates. Grover's result has been used extensively as a subroutine in many quantum algorithms, for examples see~\cite{Brassard_1998,D_rr_2006,durr1996quantum}. We show how to reduce the average number of additional gates per oracle query while keeping the number of oracle queries as close to the optimum as we wish. We also prove that our algorithm is optimal up to a multiplicative constant.

    \subsection{Prior work} 
        Since the invention of Grover's algorithm, there were several attempts to improve it further.
        In \cite{groverfast} the author improves the number of non-oracle quantum gates. Using a simple pattern of small diffusion operators the following result is obtained.
        
        \begin{thm}[\cite{groverfast}]
        For every \( \alpha > 2 \) and any sufficiently large \(N\) there exists a quantum algorithm that finds the unique marked element among \(N\) with probability tending to \(1\), using fewer than \(\frac{\pi}{4}\sqrt{N}\paren*{\frac{1}{1-{(\log_2{N})}^{1 - \alpha}}}\) oracle queries and no more than \(\frac{9}{8} \pi \alpha\sqrt{N}\log_{2}{\log_{2}{N}}\) non-oracle gates.
        \end{thm}
    
        Later, in \cite{wolfsearch} the authors reduce the number of non-oracle gates even further.

        \begin{thm}[\cite{wolfsearch}]
        For any integer \(r >0\) and sufficiently large \(N\) of the form \(N = 2^{n}\), there exists a quantum algorithm that finds the unique marked element among \(N\) with probability \(1\), using \((\frac{\pi}{4}+o(1))\sqrt{N}\) queries and \(\calO(\sqrt{N}\log^{r}{N})\) gates.
        For every \(\varepsilon > 0\) and sufficiently large \(N\) of the form \(N = 2^{n}\), there exists a quantum algorithm that finds the unique marked element among \(N\) with probability \(1\), using \(\frac{\pi}{4}\sqrt{N}\paren*{1+\varepsilon}\) queries and \(\calO\paren*{\sqrt{N}\log\paren*{\log^{*}N}}\) gates.
        \end{thm}
        
        In the same paper, the authors raise questions regarding removing the $\log\paren*{\log^{*}N}$ factor in gate complexity, which we answer in the affirmative in \Cref{thm:main}, and dealing with oracles that mark multiple elements.
        Note that both aforementioned results assume that the given oracle marks only a single element.
        
        The concept of benefits arising from the use of local diffusion operators has been studied in other papers, e.g. \cite{zhang}.
    \subsection{Our results}
        We present an algorithm which uses only $\calO(\sqrt{N})$ non-oracle gates while making only $\calO(\sqrt{N})$ oracle queries. Additionally, to remedy the objections against optimizing the average number of additional elementary gates per oracle query mentioned in~\cite{wolfsearch}, we introduce the concept of \emph{partial uncompute} --- a technique that achieves asymptotical improvement in the total number of elementary gates in many combinatorial problems, such as Unique \(k\)-SAT (see e.g.~\cite{uniquesat} for the definition of Unique \(k\)-SAT). The high-level idea of the technique is to utilize the structure of the given oracle and store some intermediate information on ancilla qubits when implementing the oracle. If between two consecutive oracle queries we applied elementary gates only on a small number of qubits, we expect that the most of intermediate information has not changed at all. Leveraging this phenomenon, we can reduce the asymptotic number of gates needed to implement the circuit.
        
        In Grover's algorithm the diffusion operator is applied on $\calO(\log N)$ qubits, so we cannot benefit from partial uncompute. We need to have an algorithm that on average affects only a small subset of qubits between consecutive oracle queries. To handle this problem we introduce an algorithm for generating quantum circuits that drastically reduces the average number of additional gates. The algorithm can be used to generate circuits that work for any number of qubits and can be potentially implemented on NISQ devices. Moreover, the algorithm improves on the results of \cite{groverfast} and \cite{wolfsearch} and can be summarized as follows.
    
        \begin{thm}\label{thm:main}
            Fix any \(\varepsilon \in (0, 1)\), and any \(N \in \mathbb{N}\) of the form $N = 2^n$. Suppose we are given a quantum oracle \(O\) operating on \(n\) qubits that marks exactly one element. Then there exists a quantum circuit \(\calA\) which uses the oracle \(O\) at most
            \(\paren*{1 + \varepsilon} \frac{\pi}{4} \sqrt{N}\) times
            and uses at most
            \(\calO \paren*{\log\paren*{\sfrac{1}{\varepsilon}} \sqrt{N}}\)
            non-oracle basic gates, which finds the element marked by \(O\) with certainty.
        \end{thm}
        It is important to note that the constant hidden by \(\calO\) notation in \Cref{thm:main} is independent of both \(N\) and \(\varepsilon\).
        Moreover, any quantum algorithm tackling this problem must perform at least $\frac{\pi}{4}\sqrt{N}$ oracle calls, see~\cite{zalka}.

        The algorithm $\calA$ can be, in broad strokes, explained as follows. We build a quantum circuit recursively according to some simple rules. The resulting circuit concentrates enough amplitude in the marked element. After that, we apply Amplitude Amplification \cite{ampamp} to it. The main idea in $\calA$ is to explore small diffusion operators (diffusion operators applied on a small subset of qubits). They are obviously easier to implement than large ones and require fewer elementary gates. Moreover, if they are applied wisely, they can be extremely efficient in concentrating amplitude in the marked element.

        If we combine the partial uncompute technique with \cref{thm:main} to solve a Unique \(k\)-SAT problem, we get the following corollary.

\begin{restatable}{corollary}{ksat}
        Consider the Unique \(k\)-SAT problem with $n$ variables and $c$ clauses.
        There exists a quantum circuit that uses 
        $\calO(c\log(c)2^{n/2}/n)$
        total (oracle and non--oracle) gates and solves the problem with certainty.
\end{restatable}

        \noindent
        It worth mentioning that it is a slight improvement over the na\"ive application of Grover's algorithm to solve the Unique \(k\)-SAT problem, because Grover's algorithm requires $\calO((n + c) \cdot 2^{n/2})$ elementary gates to solve the problem with certainty.
    
        By result of~\cite{zalka}, the optimal number of queries to the oracle required for solving unstructured search problem with certainty is $\frac{\pi}{4} \sqrt{N}$. We show that the trade-off between the number of oracle queries and non-oracle gates from \cref{thm:main} is optimal up to a constant factor.
        
        \begin{corollary}\label{corollary:optimality} 
            There exists a number $\delta>0$ such that for any $\varepsilon \in (0, 1)$ and for any quantum circuit \(\calA\) the following holds.
            If \(\calA\) uses at most \( \delta\log\paren*{\sfrac{1}{\varepsilon}} \sqrt{N}\) non-oracle gates and finds the element marked by \(O\) with certainty then \(\calA\) uses the oracle \(O\) at least \(\paren*{1 + \varepsilon} \frac{\pi}{4} \sqrt{N}\) times.
        \end{corollary}
        
        Last but not least, following the approach of \cite{valiant}, we asymptotically reduce the overhead incurred when reducing the unstructured search problem with multiple marked elements to the unstructured search problem with exactly one marked element. We modify the oracle in a classical randomized way so that the modified oracle marks exactly one element with constant probability. This is achieved by randomly choosing an affine hash function that excludes some elements from the search space. If the number of marked elements $K$ is known in advance, we will sample a hash function from such a set so that the expected number of marked elements after combining the oracle with the function is equal to one. We formulate this result as the following theorem.
        
        \begin{restatable}{thm}{multipoint}\label{thm:multipoint}
            Let \(N \in \mathbb{N}\) be of the form \(N = 2^{n}\). Assume that we are given a phase oracle \(O\) that marks \(K\) elements, and we know the number \(k\) given by \(k = 1 + \ceil*{\log_2{K}}\).
            Then one can find an element marked by \(O\) with probability at least \(\frac{1}{16}\), using at most \(\calO \paren*{\sqrt{\frac{N}{K}}}\) oracle queries and at most \(\calO \paren*{\log K \sqrt{\frac{N}{K}}}\) non--oracle basic gates.
        \end{restatable}
        
        What is more: we can extend this approach to the case when the number of marked elements is unknown by trying different values of $K$ and applying the same algorithm. This can be done in such a way that the number of oracle queries and the average number of additional elementary gates per oracle query are asymptotically the same as in the case of known $K$.
        
    \subsection{Further remarks}
    
    While our results describe asymptotic behavior, the techniques used to achieve them are quite practical. As described in \cite{patent}, they may be applicable for achieving the improvements in implementations of unstructured search on existing and near-future NISQ devices. The previous implementations of unstructured search beyond spaces spanned by 3-qubits were unsuccessful \cite{ibm},
    perhaps techniques described here can allow searching larger spaces on current hardware.
    
    \subsubsection*{Organization}
    In \cref{section:premilinaries} we briefly discuss the computational model and notation used throughout this paper. In \cref{section:circuit} we describe our main algorithm for constructing quantum circuits. Next, in \cref{section:lowerbound} we prove that our algorithm is optimal (up to a constant factor) in the number of additional elementary gates. Later, in \cref{section:uncompute} we introduce the partial uncompute technique and show an example application to a hard combinatorial problem. Finally, in \cref{section:multipoint} we proceed to reduce the unstructured search problem with multiple marked elements to the unstructured search problem with one marked element.

\section{Preliminaries}
\label{section:premilinaries}
    In the \emph{unstructured search problem} we are given a function $f\colon \{0, 1\}^n \to \{0, 1\}$ for some $n \in \mathbb{N}$ and we wish to find $x \in \{0, 1\}^n$ such that $f(x) = 1$. We will call such $x$ \emph{marked}. The function can be evaluated at $N$ points in total, where $N = 2^n$, and the goal is to find a marked element whilst minimizing the total number of evaluations of $f$. In the quantum version of the problem the function $f$ is given as a \emph{phase oracle} $O$, i.e. a unitary transformation given by $O \ket{x} = (-1)^{f(x)} \ket{x}$ for every computational basis vector $\ket{x}$. We still want to query $O$ the least possible number of times to find a marked element. Sometimes this problem is called the \emph{database search} problem.
    We use the standard gate model of quantum computations.
    We assume that our elementary operations are the universal set of quantum gates consisting of \(\cnot\) and arbitrary one qubit gates. We will refer to these gates as \emph{basic gates}. We note that this gate set can simulate any other universal gate set with bounded gate size with at most constant overhead, the details can be found in~\cite{nielsen}.
    
    In all following equations all operators are to be understood as applied right-to-left (i.e. as in standard operator composition), while in figures the application order is left-to-right, as is the standard when drawing quantum circuits.

    Given a positive integer \(k\), the \emph{uniform superposition state on \(k\) qubits}, denoted \(\ket{u_k}\), is defined as
    \(\ket{u_k} = \frac{1}{2^{k / 2}} \sum_{b \in \set{0, 1}^k} \ket{b}\). We extend this definition to the special case of \(k = 0\) by setting \(\ket{u_0} = 1\). 
    A useful identity which we will use throughout the derivations to come is \(\ket{u_a} \ket{u_b} = \ket{u_{a + b}}\) for \(a,b \in \mathbb{N}\).
    
    The \emph{mixing operator of size \(k\)} (alternatively also called the diffusion operator, or simply the diffuser), denoted \(G_k\), is defined as
    \(G_k = 2 \ket{u_k}\bra{u_k} - \id_{k}\),
    where \(\id_k\) is the identity matrix of size \(2^k\). From \cite{nielsen} we know that we can implement \(G_k\) using \(\calO(k)\) basic gates (and this is best possible). 
    
    To prove optimality of our results and to define the partial uncompute technique we consider what happens when operators do not act on some subset of qubits. Intuitively, it means that we do not need to use these qubits when implementing this operator using basic gates. 
    We say that a unitary matrix \(A\) operating on \(n\) qubits (here denoted \(\set{q_1, \dots, q_n}\)) \emph{does not act on the qubit \(q_i\)} if 
    \[
        A = \Swap(q_i, q_n)(A'\otimes \id_1 )\Swap(q_i, q_n)
    \]
    where \(A'\) is some unitary matrix operating on \(n - 1\) qubits, and \(\Swap(a,b) = \cnot(a,b) \cnot(b,a) \cnot(a,b)\).
    Otherwise we say that \(A\) \emph{acts} on qubit \(q_i\).
    We say that operator \(A\) \emph{may act} on qubits \(q_{i_1}, \ldots, q_{i_m}\) if it does not act on qubits \(\{ q_1, \ldots, q_n\}\setminus \{q_{i_1},  q_{i_2}, \ldots, q_{i_m}\} \).
    
    In the proof of \cref{thm:main} we will need the following result from \cite{ampamp}, which we will refer to as \emph{Amplitude Amplification}.
    \begin{thm}[\cite{ampamp}, p.7 Theorem 2]\label{thm:badampamp}
        Let \(\mathcal{A}\) be any quantum algorithm operating on \(n\) qubits that uses no measurements, and let \(f \colon \set{0, 1}^{n} \rightarrow \set{0,1}\) be any boolean function with a corresponding phase oracle \(O\). Let \(a\) be the probability that measuring \(\mathcal{A} \ket{00\ldots 0}\) yields \(\ket{t}\) such that \(f(t) = 1\), and assume that \(a \in (0, 1)\). Let \(\theta \in (0, \pi/2)\) be such that \(\paren*{\sin \theta}^2 = a\), and let \(s = \floor*{\frac{\pi}{4\theta}}\).
        Then measuring \( \paren*{-\mathcal{A}F_0\mathcal{A}^{\dagger}O}^{s} \mathcal{A} \ket{00 \ldots 0}\) yields \(\ket{t}\) such that \(f(t) = 1\) with probability at least \(\max \set{1 - a, a}\), where
        \[
            F_0 \ket{t} = \begin{cases}
                \ket{t}, & \ \text{ if } \ t \neq 0\\
                -\ket{t}, & \ \text{ if } \ t = 0 \text{.}
            \end{cases}
        \]
    \end{thm}
    Note that this result requires us to know the value of \(a\) \emph{precisely}. However, this is not a problem for us, as we shall later see.
    
    There is a simple corollary one can obtain from the proof of \cref{thm:badampamp} (it is noted as Theorem 4 in~\cite{ampamp}, however the authors do not make the constants explicit in their formulation). The precise formula one gets for the probability of success when measuring \(\paren*{-\mathcal{A}F_0 \mathcal{A}^{\dagger} O}^m \mathcal{A} \ket{00 \ldots 0}\) is in fact equal to \(\sin^2 \paren*{\paren*{2 m + 1} \theta}\).
    If it were to happen that \(r = \pi/\paren*{4 \theta} - {1}/{2}\) was an integer, then we could simply set the number of iterations to \(r\) and obtain a solution with certainty. Now it remains to note that we can easily modify \(\mathcal{A}\) to lower \(\theta\) slightly so that the new value of \(r\) is indeed an integer.
    It is important for our results that the number of iterations is in fact bounded by \(\floor*{\frac{\pi}{4 \theta}} + 1\), which is formulated as the theorem below.
    \begin{thm}[\cite{ampamp}, Theorem 4 restated]\label{thm:ampamp}
        Let \(\mathcal{A}\) be any quantum algorithm operating on \(n\) qubits that uses no measurements, and let \(f \colon \set{0, 1}^{n} \rightarrow \set{0,1}\) be any boolean function. Let \(a\) be the probability that measuring \(\mathcal{A} \ket{00\ldots 0}\) yields \(\ket{t}\) such that \(f(t) = 1\), and assume that \(a \in (0, 1)\). Let \(\theta \in (0, \pi/2)\) be such that \(\paren*{\sin \theta}^2 = a\).
        
        Then there exists a quantum algorithm that uses \(\mathcal{A}\) and \(\mathcal{A}^{\dagger}\) at most \(\floor*{\frac{\pi}{4 \theta}} + 2\) times each, which upon measurement yields \(\ket{t}\) such that \(f(t) = 1\) with certainty.
    \end{thm}

    Note that the bound \(\floor*{\frac{\pi}{4 \theta}} + 2\) follows from the extra \(\mathcal{A}\) applied at the beginning of the Amplitude Amplification (as we are counting the applications of \(\mathcal{A}\) and \(\mathcal{A}^{\dagger}\) and not iterations).
\section{Structure of the \(W_m\) circuit}
\label{section:circuit}
    \begin{ddef}\label{def:Wojter}
        Let $\overline{k} = (k_1, \dots, k_m)$ be a sequence of positive integers and let $n := \sum_{j=1}^m k_j$. Given a quantum oracle $O$, for $j \in \set{0, \dots, m}$ we define the circuit $W_j$ recursively as follows:
        \begin{multline*}
            W_0 := \id_{n},
            \\
            W_j :=
                W_{j - 1} \cdot \paren*{\id_{k_1 + \dots + k_{j - 1}} \otimes G_{k_j} \otimes \id_{k_{j + 1} + \dots + k_m}} \\
                \cdot W_{j - 1}^{\dagger} 
                \cdot O \cdot W_{j - 1}, \quad j\in\set{1,2,\dots,m}.
        \end{multline*}
    \end{ddef}
    For an example of how the circuits \(W_m\) look like see \cref{fig:example,fig:scheme}. Observe that the circuit \(W_m\) uses the oracle \(O\) exactly \(\paren*{3^m - 1} / 2\) times. Moreover, as the mixing operator \(G_k\) can be implemented using \(\calO \paren*{k}\) basic quantum gates, for a given \(\overline{k}\), one can implement \(W_m\) for these diffuser sizes using \(\calO \paren*{\sum_{j = 1}^m \ k_j 3^{m - j}}\) basic quantum gates, not including the gates necessary for the implementation of the oracle.
    
    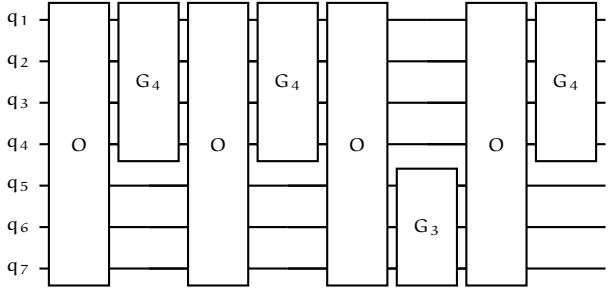
\begin{figure}[H]
    {\scriptsize
        \begin{quantikz}[row sep=0.1cm, column sep=0.125cm]
            \lstick{$q_1$} & \gate[7,disable auto height][0.8cm]{O} & \gate[4,disable auto height][0.8cm]{G_4} & \gate[7,disable auto height][0.8cm]{O} & \gate[4,disable auto height][0.8cm]{G_4} & \gate[7,disable auto height][0.8cm]{O} & \qw           & \gate[7,disable auto height][0.8cm]{O} & \gate[4,disable auto height][0.8cm]{G_4} & \qw \\
            \lstick{$q_2$} & \ghost{O}          & \qw           & \qw                & \qw           & \qw                & \qw           & \qw                & \qw           & \qw \\
            \lstick{$q_3$} & \ghost{O}          & \qw           & \qw                & \qw           & \qw                & \qw           & \qw                & \qw           & \qw \\ 
            \lstick{$q_4$} & \ghost{O}          & \qw           & \qw                & \qw           & \qw                & \qw           & \qw                & \qw           & \qw \\
            \lstick{$q_5$} & \ghost{O}          & \qw           & \qw                & \qw           & \qw                & \gate[3,disable auto height][0.8cm]{G_3} & \qw                & \qw           & \qw \\ 
            \lstick{$q_6$} & \ghost{O}          & \qw           & \qw                & \qw           & \qw                & \qw           & \qw                & \qw           & \qw \\
            \lstick{$q_7$} & \ghost{O}          & \qw           & \qw                & \qw           & \qw                & \qw           & \qw                & \qw           & \qw \\
        \end{quantikz}
        \caption{$W_2$ for $\overline{k} = (4, 3)$}
        \label{fig:example}
    }
    \end{figure}

    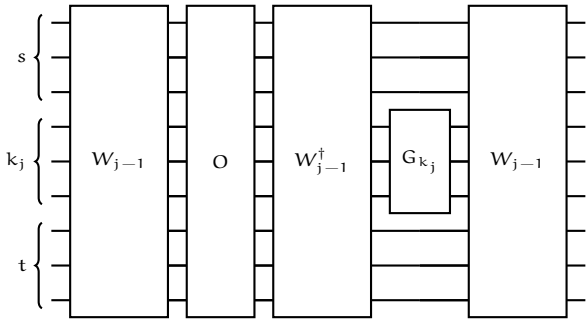
\begin{figure}[H]
    {\scriptsize
        \begin{quantikz}[row sep=0.01cm, column sep=0.25cm]
           \lstick[3]{$s$}     & \gate[wires=9,disable auto height][1.3cm]{W_{j-1}} & \gate[9,disable auto height][0.9cm]{O}   & \gate[9,disable auto height][1.3cm]{W_{j-1}^{\dagger}} & \qw               & \gate[9,disable auto height][1.3cm]{W_{j-1}} & \qw \\
            \lstick{}          & \ghost{W}                      & \qw                  & \qw                                & \qw               & \qw                      & \qw \\
            \lstick{}          & \ghost{W}                      & \qw                  & \qw                                & \qw               & \qw                      & \qw \\ 
            \lstick[3]{$k_j$}  & \ghost{W}                      & \qw                  & \qw                                & \gate[3,disable auto height]{G_{k_j}} & \qw                      & \qw \\
            \lstick{}          & \ghost{W}                      & \qw                  & \qw                                & \qw               & \qw                      & \qw \\ 
            \lstick{}          & \ghost{W}                      & \qw                  & \qw                                & \qw               & \qw                      & \qw \\
            \lstick[3]{$t$}    & \ghost{W}                      & \qw                  & \qw                                & \qw               & \qw                      & \qw \\
            \lstick{}          & \ghost{W}                      & \qw                  & \qw                                & \qw               & \qw                      & \qw \\
            \lstick{}          & \ghost{W}                      & \qw                  & \qw                                & \qw               & \qw                      & \qw \\
        \end{quantikz}
        }
        \caption{Graphical representation of the \(W_j\) circuit. Note that the oracle in \(W_{j - 1}\) manipulates all of the qubits, however no other gate does so. In this picture $s = k_1 + \dots + k_{j-1}$ and $t = k_{j+1} + \dots + k_m$.}
        \label{fig:scheme}
    \end{figure}

    \subsection{Obtaining the recurrence for amplitude in the target}
    
    In this subsection we aim to derive a recurrence formula that will allow us to compute the amplitude our circuit \(W_m\) concentrates in the unique marked state. We assume we are given a phase oracle \(O\) operating on \(n\) qubits, that marks a single state denoted \(\target\). We have also fixed a vector of positive integers \(\overline{k} = (k_1, \dots, k_m)\), such that \(k_1 + \dots + k_m = n\).
    For the duration of this section, we introduce the following notational conveniences.
    We split the marked state \(\ket{\target}\) according to \(\overline{k}\) as
    \[
        \ket{\target} = \ket{\target_1} \ket{\target_2} \dots \ket{\target_m}
    \]
    where \(\target_1\) consists of bits of \(\target\) numbered 1 to \(k_1\), \(\target_2\) of the bits numbered \(k_1 + 1\) to \(k_1 + k_2\) etc.
    Moreover, for given \(i, j\) we define the following product
    \[
        \ket{\target_i^j} = \ket{\target_i} \ket{\target_{i + 1}} \dots \ket{\target_j}\text{.}
    \]
    If the interval \([i, j]\) happens to be empty, we understand \(\ket{\target_i^j}\) to be the scalar \(1\).
    To shorten the derivations about to follow, we will also use these shorthands
    \[
        \ket*{\overline{\target}_j} = \frac{1}{2^{k_j / 2}} \sum_{\substack{b \in \set{0, 1}^{k_j} \\ b \neq \target_j}} \ket{b}\text{,}
    \]
    \[
        \ket{u_{1}^{j}} = \ket{u_s}
    \]
    where \(s = k_1 + \dots + k_j\) with the additional convention that \(\ket{u_1^0} = 1\). Observe that we have the equations \(\ket{u_{k_j}} = \ket*{\overline{\target}_j} + 2^{-k_j / 2} \ket{\target_j}\) and \(\braket{\target_j}{ \overline{\target}_j} = 0\).
    
    We begin by introducing two simple lemmas.
    \begin{lem}\label{lem:orthogonal}
        Fix any \(m \in \mathbb{N}_{+}\), and any \(\overline{k} = (k_1, \dots, k_m) \in \mathbb{N}_{+}^m\), and let \(n = \sum_{j=1}^m k_j\). Assume that we are given a phase oracle \(O\) that operates on \(n\) qubits and marks a single vector of the standard computational basis denoted \(\target\).
        Then for any \(j \in \set{0, \dots, m - 1}\), and any vector \(\ket{\phi} \in \paren*{\mathbb{C}^2}^{\otimes t}\) (where \(t = k_{j + 1} + \dots + k_{m}\)) such that \(\braket{\phi}{\target_{j+1}^{m}} = 0\) we have
        \[
            W_{j} \paren*{\ket{u_1^{j}} \ket{\phi}} = \ket{u_1^{j}} \ket{\phi}\text{.}
        \]
    \end{lem}
    \begin{proof}
      
        Observe, that as \(\braket{\phi}{\target_{j+1}^{m}} = 0\) the vector \(\ket{u_s} \ket{\phi}\) is an eigenvector of the operator \(O\) with eigenvalue \(1\). Thus the lemma's assertion will be proved, if we show that it is also an eigenvector (with eigenvalue \(1\)) of each diffusion operator that appears in \(W_j\), that is \(\paren*{\id_{a} \otimes G_{b} \otimes \id_{n - b - a}}\ket{u_1^j}\ket{\phi} = \ket{u_1^j}\ket{\phi}\) whenever \(a + b \leq k_1 + k_2 + \dots + k_j\), which we quickly verify by the direct calculation below.
        
        \begin{multline*}
            \paren*{\id_{a} \otimes G_{b} \otimes \id_{n - b - a}} \paren*{\ket{u_1^j} \ket{\phi} } 
            \\
            = \paren*{\id_{a} \otimes G_{b} \otimes \id_{n - b - a}} \paren*{\ket{u_{a}} \ket{u_{b}} \ket{u_{k_1 + \dots + k_j - a - b}} \ket{\phi} }
            \\
            = \paren*{\id_{a}\ket{u_{a}}} \paren*{G_{b} \ket{u_{b}}} \paren*{\id_{n - b - a} \ket{u_{k_1 + \dots + k_j - b - a}} \ket{\phi} }
            \\
            = \ket{u_{a}} \ket{u_{b}} \ket{u_{k_1 + \dots + k_j - b - a}} \ket{\phi} = \ket{u_1^j} \ket{\phi}
        \end{multline*}
    \end{proof}
    
    \begin{lem}\label{lem:nooverlap}
        Fix any \(m \in \mathbb{N}_{+}\), and any \(\overline{k} = (k_1, \dots, k_m) \in \mathbb{N}_{+}^m\), and let \(n = \sum_{j=1}^m k_j\). Assume that we are given a phase oracle \(O\) that operates on \(n\) qubits and marks a single vector of the standard computational basis denoted \(\target\).
        Then for any \(j \in \set{1, \dots, m}\) we have
        \begin{multline*}
            W_{j-1} \paren*{\id_{s - k_j} \otimes G_{k_j} \otimes \id_{n - s}} W_{j - 1}^{\dagger} \ket{\target} \\
            = \paren*{\frac{2}{2^{k_j}} - 1} \ket{\target} + \ket{\vartheta}
        \end{multline*}
        where \(s = k_1 + \dots + k_{j}\), and \(\ket{\vartheta}\) is some state orthogonal to \(\ket{\target}\).
    \end{lem}
    \begin{proof}
        Observe that each diffusion operator in \(W_{j - 1}\) (and thus also in \(W_{j - 1}^{\dagger}\)) operates on the qubits numbered \(\set{1, \dots, k_1 + \dots + k_{j - 1}}\), thus there exists a vector \(\ket{\eta} \in \paren*{\mathbb{C}^{2}}^{\otimes \paren*{k_1 + \dots + k_{j - 1}}}\) such that 
        \[
            W_{j - 1} ^{\dagger} \ket{\target} = \ket{\eta} \ket{\target_{j}^{m}}\text{.}
        \]
        Equipped with this observation, we proceed to directly compute the desired result
        \begin{multline*}
            W_{j-1} \paren*{\id_{s - k_j} \otimes G_{k_j} \otimes \id_{n - s}} W_{j - 1}^{\dagger} \ket{\target} 
            \\
            = W_{j-1} \paren*{\id_{s - k_j} \otimes G_{k_j} \otimes \id_{n - s}} \ket{\eta} \ket{\target_{j}^{m}}
            \\
            =W_{j - 1} \paren*{\id_{s - k_j} \otimes G_{k_j} \otimes \id_{n - s}} \ket{\eta} \ket{\target_j} \ket{\target_{j+1}^{m}} 
            \\
            = W_{j - 1} \paren*{\ket{\eta} \paren*{G_{k_j} \ket{\target_j}} \ket{\target_{j + 1}^{m}}}
            \\
            = W_{j - 1} \paren*{
            \ket{\eta} \paren*{\frac{2}{2^{k_j / 2}}\ket{u_{k_j}} - \ket{\target_{j}}} \ket{\target_{j + 1}^{m}}
            }
            \\
            = W_{j - 1} 
            \paren*{\frac{2}{2^{k_j}} - 1} \ket{\eta} \ket{\target_{j}^{m}} 
            \\
            + 
            W_{j - 1} 
            \frac{2}{2^{k_j / 2}} \ket{\eta}\ket*{\overline{\target}_j}\ket{\target_{j + 1}^{m}}
            \\
            = \paren*{\frac{2}{2^{k_j}} - 1} \ket{\target} + \ket{\vartheta}
        \end{multline*}
        and observe that \(\ket{\vartheta}\) is orthogonal to \(\ket{\target}\), as their respective preimages under \(W_{j - 1}\) were orthogonal.
    \end{proof}
    
    \begin{lem}\label{lem:recurrence}
        Fix any \(m \in \mathbb{N}_{+}\), and any \(\overline{k} = (k_1, \dots, k_m) \in \mathbb{N}_{+}^m\), and let \(n = \sum_{j=1}^m k_j\). Assume that we are given a phase oracle \(O\) that operates on \(n\) qubits and marks a single vector of the standard computational basis denoted \(\target\). Define the numbers
        \[
            \alpha_j = \bra{\target} \paren*{W_j \ket{u_1^{j}} \ket{\target_{j+1}^{m}}}
        \]
        for \(j \in \set{0, 1, \dots, m}\).
        Then \(\alpha_j\) satisfy the recurrence
        \[
        \alpha_j = \begin{cases}
            1, & \ \text{ if }\ j = 0\\
            2^{-k_j / 2} \paren*{3 - 4 \cdot 2^{-k_j}} \alpha_{j - 1}, & \ \text{ if }\ j > 0\text{.}
        \end{cases}
        \]
    \end{lem}
    
    \begin{proof}
        Clearly \(\alpha_0 = 1\) giving the base case. Now, let us assume that \(j > 0\), and we will proceed to compute \(\alpha_j\) by expanding the circuit \(W_j\) according to \cref{def:Wojter}.
        To maintain legibility we will split this computation into several steps.
        Let us define the intermediate states \(\ket{w_1}, \dots, \ket{w_5}\) by the following equations
        
        \begin{align*}
            \ket{w_1} &= W_{j-1} \paren*{\ket{u_1^j} \ket{\target_{j + 1}^{m}}}\\
            \ket{w_2} &= O \ket{w_1}\\
            \ket{w_3} &= W_{j - 1}^{\dagger} \ket{w_2}\\
            \ket{w_4} &=  \paren*{\id_{s - k_j} \otimes G_{k_j} \otimes \id_{n - s}}\ket{w_3}\\
            \ket{w_5} &= W_{j - 1} \ket{w_4}
        \end{align*}
        where \(s = k_1 + \dots + k_j\).
        \begin{widetext}
            \begin{align}
                \ket{w_1} &= W_{j - 1} \paren*{\ket{u_1^j} \ket{\target_{j + 1}^{m}}} \nonumber \\
                          &= W_{j - 1} \paren*{\frac{1}{2^{k_j / 2}} \ket{u_1^{j-1}} \ket{\target_{j}^{m}} + \ket{u_1^{j-1}} \ket*{\overline{\target}_j} \ket{\target_{j + 1}^{m}}} \nonumber \\
                          &= \frac{1}{2^{k_j / 2}} W_{j - 1} \ket{u_1^{j - 1}}\ket{\target_{j}^{m}} + \ket{u_1^{j - 1}} \ket*{\overline{\target}_j} \ket{\target_{j + 1}^{m}} \label{eqn:lemuse:1}
            \end{align}
        \end{widetext}
        Where in \cref{eqn:lemuse:1} we relied on \cref{lem:orthogonal}. Plugging this equation into the definition of \(\ket{w_2}\) we obtain
        \begin{widetext}
            \begin{align}
                \ket{w_2} &= O \paren*{\frac{1}{2^{k_j / 2}} W_{j - 1} \ket{u_1^{j - 1}}\ket{\target_{j}^{m}} + \ket{u_1^{j - 1}} \ket*{\overline{\target}_j} \ket{\target_{j + 1}^{m}}} \nonumber \\
                          &= \frac{1}{2^{k_j / 2}} \paren*{W_{j - 1} \ket{u_1^{j - 1}}\ket{\target_{j}^{m}} - 2 \alpha_{j - 1} \ket{\target}} + \ket{u_1^{j - 1}} \ket*{\overline{\target}_j} \ket{\target_{j + 1}^{m}} \label{eqn:usedef1} \\
                \ket{w_3} &= W_{j - 1}^{\dagger} \paren*{\frac{1}{2^{k_j / 2}} W_{j - 1} \ket{u_1^{j - 1}}\ket{\target_{j}^{m}} - \frac{2}{2^{k_j / 2}} \alpha_{j - 1} \ket{\target} + \ket{u_1^{j - 1}} \ket*{\overline{\target}_j} \ket{\target_{j + 1}^ {m}}} \nonumber \\
                          &= \frac{1}{2^{k_j / 2}} \ket{u_1^{j - 1}} \ket{\target_{j}^{m}} - \frac{2}{2^{k_j / 2}} \alpha_{j - 1} W_{j - 1}^{\dagger} \ket{\target} + \ket{u_1^{j - 1}} \ket*{\overline{\target}_j} \ket{\target_{j + 1}^{m}} \label{eqn:lemuse3} \\
                          &= \ket{u_1^j}\ket{\target_{j+1}^{m}} - \frac{2}{2^{k_j / 2}} \alpha_{j - 1} W_{j - 1}^{\dagger} \ket{\target} \label{eqn:defuse1} \\
                \ket{w_4} &= \id_{s - k_j} \otimes G_{k_j} \otimes \id_{n - s} \paren*{\ket{u_1^j} \ket{\target_{j + 1}^{m}} - \frac{2}{2^{k_j / 2}} \alpha_{j - 1} W_{j - 1}^{\dagger} \ket{\target} } \nonumber \\
                          &= \ket{u_1^j} \ket{\target_{j + 1}^{m}} - \frac{2}{2^{k_j / 2}} \alpha_{j - 1} \paren*{\id_{s - k_j} \otimes G_{k_j} \otimes \id_{n - s}} W_{j - 1}^{\dagger} \ket{\target} \nonumber \\
                \ket{w_5} &= W_{j - 1} \paren*{\ket{u_1^j} \ket{\target_{j + 1}^{m}} - \frac{2}{2^{k_j / 2}} \alpha_{j - 1} \paren*{\id_{s - k_j} \otimes G_{k_j} \otimes \id_{n - s}} W_{j - 1}^{\dagger} \ket{\target} } \nonumber \\
                          &= \frac{1}{2^{k_j / 2}} W_{j - 1} \ket{u_1^{j - 1}}\ket{\target_{j}^{m}} + \ket{u_1^{j - 1}} \ket*{\overline{\target}_j} \ket{\target_{j + 1}^{m}} - \frac{2}{2^{k_j / 2}} \alpha_{j - 1} W_{j - 1} \paren*{\id_{s - k_j} \otimes G_{k_j} \otimes \id_{n - s}} W_{j - 1}^{\dagger} \ket{\target} \label{eqn:long} \\
                          &= \frac{1}{2^{k_j / 2}} W_{j - 1} \ket{u_1^{j - 1}}\ket{\target_{j}^{m}} + \ket{u_1^{j - 1}} \ket*{\overline{\target}_j} \ket{\target_{j + 1}^{m}} - \frac{2}{2^{k_j / 2}} \alpha_{j - 1} \paren*{\paren*{\frac{2}{2^{k_j}} - 1} \ket{\target} + \ket{\vartheta}} \label{eqn:lemuse2}
            \end{align}
        \end{widetext}

        Note that in \cref{eqn:usedef1} we used the definition of \(\alpha_{j - 1}\), in \cref{eqn:lemuse3} we applied \cref{lem:orthogonal}, \cref{eqn:defuse1,eqn:long} follows from the definition of \(\ket*{\overline{\target}_j}\), while \cref{eqn:lemuse2} we applied \cref{lem:nooverlap}.
        Keeping in mind that \(\ket*{\vartheta}\) is orthogonal to \(\ket{\target}\) and equipped with \cref{eqn:lemuse2} we may finally compute \(\alpha_j\) as
        \begin{widetext}
        \begin{multline*}
            \alpha_j =\  \braket{\target}{w_5}
            =\  \frac{1}{2^{k_j / 2}} \bra{\target } \paren*{ W_{j - 1} \ket{u_1^{j - 1}}\ket{\target_{j}^{m}} + 2 \alpha_{j - 1} \paren*{1 - \frac{2}{2^{k_j}}} \ket{\target} } \\
            =\  \frac{1}{2^{k_j / 2}} \paren*{\alpha_{j - 1} + 2 \alpha_{j - 1} \paren*{1 - \frac{2}{2^{k_j}}}} 
            = \frac{1}{2^{k_j / 2}} \paren*{3 - 4\cdot 2^{-k_j}} \alpha_{j - 1}.
        \end{multline*}\qedhere
        \end{widetext}
    \end{proof}
    \subsection{Proof of \cref{thm:main}}
    \begin{proof}[Proof of \cref{thm:main}]
        It clearly suffices to prove the theorem under assumption that \(\varepsilon\) is small enough, let us assume that is indeed the case.
        
        Let \(\overline{k} = \paren*{k_1, \dots, k_m}\) be some sequence of positive integers to be determined later, such that \(\sum_{j = 1}^{m} k_j= n\). We will use the circuit \(W_m\) with these diffuser sizes, and utilise \cref{thm:ampamp} on top of this circuit. To estimate the number of iterations made by Amplitude Amplification, we need a precise formula for amplitude in the marked state that the circuit \(W_m H^{\otimes n}\) (the Walsh-Hadamard transform is only necessary because we assumed our circuit to be fed the state \(\ket{u_n}\), while Amplitude Amplification assumes that the state \(\ket{00\dots0}\) is the one we work with) yields --- denoted \(\alpha_m\). To this end we use the recurrence we have obtained in \cref{lem:recurrence}, to which we can provide a solution as a product\footnote{It is interesting to note, that setting each \(k_j = 2\) yields \(\alpha_m = 1\) in which case Amplitude Amplification is not necessary, thus  giving a simple algorithm solving the unstructured search problem with each diffuser size bounded by a constant. However, the number of oracle queries it makes is \(\calO(3^{n / 2})\).}
        \begin{align}\label{req:solexact}
            \alpha_m &= \prod_{j = 1}^{m} \paren*{ 2^{-k_j / 2} \paren*{3 - 4 \cdot 2^{-k_j}}  } \nonumber \\
                     &= 2^{-n / 2} \prod_{j = 1}^{m} \paren*{3 - 4 \cdot 2^{-k_j}} \nonumber \\
                     &= 2^{-n / 2} \cdot 3^{m} \prod_{j=1}^{m}\paren*{1 - \frac{4}{3} \cdot 2^{-k_j}}.
        \end{align}
        Let us now consider the case of particular choice of \(\overline{k}\), namely \(k_j = (x + 1) j\), where \(x \in \mathbb{N}_{+}\) is some fixed constant. We will for now assume, for the sake of simplicity, that the number of qubits \(n\) is precisely equal to \((x + 1) + 2(x + 1) + \dots + m(x + 1) = (x + 1) m (m + 1) / 2\). We will later argue that this assumption is not necessary.
        Observe that in particular we have
        \begin{equation}\label{eqn:depthbound}
            m \in \Theta \paren*{\sqrt{n / x}}\text{.}
        \end{equation}

        Thus we can lower bound the product in \(\alpha_m\) as follows:
        \[
            \prod_{j = 1}^{m} \paren*{1 - \frac{4}{3} \cdot 2^{-(x + 1) j}} \geq \prod_{j=1}^{m} \paren*{1 - 2^{-xj}} \geq \prod_{j=1}^{\infty} \paren*{1 - 2^{-xj}}\text{.}
        \]
        We recall the beautiful identity due to Euler \cite{euler}, which relates the infinite product on right hand side with pentagonal numbers
        \[
            \prod_{j=1}^{\infty} \paren*{1 - z^j} = 1 + \sum_{j=1}^{\infty}\paren*{-1}^{j} \paren*{z^{(3j - 1)j / 2} + z^{(3j + 1)j/2}}
        \]
        which we use to lower bound the product for \(z \in [0, 1)\) as
        \[
            \prod_{j=1}^{\infty} \paren*{1 - z^j} \geq 1 - z - z^2
        \]
        by grouping latter terms in the series in consecutive pairs and observing that each such pair has a positive sum. This gives us the inequality
        \begin{equation}\label[ineq]{ineq:alpha}
            \alpha_m \geq 2^{-n / 2} \cdot 3^m \cdot \paren*{1 - 2^{-x} - 2^{-2x}}\text{.}
        \end{equation}
        
        Using \cref{thm:ampamp}, we need at most
        \[
            \floor*{\frac{\pi}{4 \theta_m}} + 2
        \]
        applications of our circuit \(W_m\) and its conjugate, where \(\theta_m = \arcsin{\alpha_m}\). Using the standard inequality for \(z \in (0, 1]\)
        \[
            \sin{z} \leq z
        \]
        which we can restate as
        \begin{equation}\label[ineq]{ineq:arcsin}
            \frac{1}{\arcsin{z}} \leq \frac{1}{z}\text{.}
        \end{equation}
        
        \Cref{ineq:alpha,ineq:arcsin} together imply that the number of applications of \(W_m\) and \(W_m^\dagger\)  in Amplitude Amplification is bounded by
        \[
            \frac{\pi}{4} \cdot \frac{1}{1 - 2^{-x} - 2^{-2x}} \cdot 2^{n /2 }\cdot 3^{-m} + 2\text{.}
        \]
        Observe that each \(W_m\) (and thus also \(W_m^{\dagger}\)) uses \(\paren*{3^m - 1} / 2\) oracle calls. Thus the total number of oracle calls is bounded by 
        \[
            \frac{\pi}{4} \cdot \frac{1}{1 - 2^{-x} - 2^{-2x}} \cdot 2^{n / 2} + 2 \cdot 3^{m} - 2
        \]
        thus, we are only a factor of \(\frac{1}{1 - 2^{-x} - 2^{-2x}}\) away from optimal number of oracle calls, as by \cref{eqn:depthbound} the additive term is negligible.
        
        Let us count the number of \emph{non-oracle} gates used by our algorithm. Note that the overhead of operations used by Amplitude Amplification other than applications of \(W_m\) is negligible compared to the cost of the \(W_m\) circuit. Each \(W_m\) can be implemented using
        \[
            \calO \paren*{\sum_{j=1}^{m}xj \cdot 3^{m - j}}
        \]
        non-oracle gates, giving us at most
        \begin{multline*}
            \calO \paren*{ \paren*{\sum_{j = 1}^{m} xj \cdot 3^{m - j}} 2^{n / 2} \cdot 3^{-m}} 
            \\
            = \calO \paren*{x \cdot 2^{n / 2} \cdot \sum_{j = 1}^{m} j 3^{-j} } = \calO \paren*{x \cdot 2^{ n / 2}}
        \end{multline*}
        non-oracle gates used by the entire algorithm. Now setting \(x \in \Theta\paren*{\log\paren*{\varepsilon^{-1}}}\) concludes the proof in this special case.
    
        Now we briefly explain how to deal with arbitrary number of qubits.
        We wish to get a suitable sequence \(\overline{k}\) for a specific positive integers \(x\) and \(n\). We do it as follows:
        let \(m = \max \set{k \ \colon \ \sum_{j \leq k} (x + 1)j \leq n}\), and define for \(j \in \set{1, \dots, m}\)
        \[
            k_j = \begin{cases}
                (x + 1)j & \quad \text{ if } \ j < m\\
                n - \sum_{k < m} (x+1)k & \quad \text{ if }\ j = m\text{.}
            \end{cases}
        \]
        By the choice of \(m\), we easily get that \(k_m \in [(x + 1)m, 3(x + 1)m)\). Observe that the number of gates necessary to implement \(W_m\) goes up by a factor of at most \(3\), thus that part of the calculation does not change.
        Next, observe that in \cref{req:solexact}, the final expression is monotonely increasing in \(k_j\), thus our lower bound in \cref{ineq:alpha} still holds. Thus, further analysis also does not change, concluding the proof.
    \end{proof}
    
    \begin{remark}
        The above analysis could be generalised to the setting of underlying space being decomposable into a tensor product as
        \[
            H_1 \otimes H_2 \otimes \dots \otimes H_m
        \]
        where of course the time complexity of the algorithm will depend on the relative dimensions of \(H_j\). However, this would not improve the proof's clarity, and does not really provide a significantly wider scope of applications, so we refrain from including it.
    \end{remark}
    \section{Optimality}
\label{section:lowerbound}

In this section we show the following lower bound for the number of oracle queries.

\begin{thm}\label{thm:lowerbound}
    Fix \(p\in(0,1)\), \(n \in \mathbb{N}\) and \(N = 2^{n}\).
    Let \(T=T(N,p)\) be the number of oracle queries in the optimal (i.e., minimizing the number of oracle queries) search algorithm that is needed to find the marked element with probability at least \(p\).
    There exists a constant \(C>0\), which possibly depends on \(p\) but does not depend on \(N\), such that for any \(\eta>0\) and any algorithm \(\calA\) the following holds.
    If \(\calA\) uses at most \(\eta T\) additional basic gates and finds the marked element with probability at least \(p\) then \(\calA\) needs to query oracle at least \(T+\floor*{2^{-C\eta}T}\) times.
\end{thm}

As a byproduct we reprove the Zalka's estimation from \cite{zalka} (\cref{cor:Zalka}) and at the end of the section we shortly explain how the above theorem implies \cref{corollary:optimality}.

Let \(m\geq n\) be the number of qubits which we use.
Assume that we have at our disposal a phase oracle \(O^y\) operating on \(n\) qubits with one marked element \(y\).
Any quantum algorithm that solves unstructured search problem has the following form: we start with some initial quantum state \(\ket{s}\) and apply the alternating sequence of oracle queries \(O^y\) and unitary operators \(U_1,\ldots,U_R\) (each of which acts on \(m\) qubits).
Thus as a result we get a state
\[
    \ket{t}=U_RO^yU_{R-1}O^y\ldots U_{1}O^y\ket{s}\text{.}
\]
It is convenient to investigate the algorithm's behavior for all possible \(y\in\{0,1\}^n\) simultaneously.
For this purpose we consider the following sphere and its subset.
Let
\[
    S=\left\{z\in\left((\mathbb{C}^2)^{\otimes m}\right)^N: |z|=\sqrt N\right\}
\]
and
\[
    \hat S=\{(z_1,\ldots,z_N)\in S:z_1=\ldots=z_N\}\text{.}
\]
Let \(y_j\) for \(j\in\set{1,\ldots,N}\) be a sequence of all elements of \(\{0,1\}^n\).
We use the following two actions of unitary group on the sphere \(S\).
For \(U \in {\rm U}(2^{m})\) and \(z\) in \(S\) we put:
\[
    Uz=(Uz_1,Uz_2,\ldots,Uz_N)
\]
and
\[
    O_Uz=(UO^{y_1}z_1,UO^{y_2}z_2,\ldots,UO^{y_N}z_N)\text{.}
\]
By straightforward calculations we get the following observation.
\begin{observation}\label{observation:dwa}
    For \(z\) in \(\hat S\) we have
    \[
        |O_Uz-Uz|=2\text{.}
    \]
\end{observation}

We consider the following sequences of points on the sphere \(S\).
\[
    \tilde s=(\ket{s},\ldots,\ket{s}) \text{ and } s_{i}= O_{U_R}\ldots O_{U_{i+1}}U_i\ldots U_1\tilde s\text{.}
\]
Let us recall the inequality proved in \cite{inbook} (see also \cite{zalka}) which is crucial for our considerations.
\begin{lem}\label{lem:crucial}
    If the algorithm finds marked element with probability at least \(p\), then
    \begin{equation}\label{crucialinequality}
        |s_R-s_0|^2\geq h(p)\text{,}
    \end{equation}
    where \(h\) is a function given by the formula
    \[
        h(p)=2N-2\sqrt N \sqrt p-2\sqrt N \sqrt{N-1}\sqrt{1-p}\text{.}
    \]
\end{lem}

The advantage of working on the sphere is that the distance between points on the sphere \(S\) is connected with the angle between them.
For \(a,b\in S\) let \(\varphi_{a,b}\) be the angle between them, i.e.
\begin{equation}\label{eqn:angle_def}
    \varphi_{a,b} = 2 \arcsin\paren*{\frac{|{a - b}|}{2\sqrt{N}}} \in [0, \pi]\text{.}
\end{equation}
Such angle is proportional to the length of the shortest arc on \(S\) connecting \(a\) and \(b\), so in particular it satisfies triangle inequality:
\[
    \varphi_{a,b} + \varphi_{b,c} \geq \varphi_{a,c}\text{.}
\]
Put
\begin{equation}\label{eqn:defalpha}
    \alpha=2\arcsin{\left(1/\sqrt N\right)}\text{.}
\end{equation}
Now let us consider distances between elements of sequence \(s_i\).

\begin{observation}\label{observation:katalfa}
    For \(i \in \set{0,\ldots,R-1}\) we have
    \[
        |s_i-s_{i+1}|=2 \ \hbox{ and } \ \varphi_{s_i,s_{i+1}}=\alpha\text{.}
    \]
\end{observation}

\begin{proof}
    By \cref{observation:dwa} we have
    \[
        |s_i-s_{i+1}|=|O_Uz-Uz|=2\text{,}
    \]
    where \(U=U_{i+1}\) and \(z=U_i\ldots U_1s_0\).
    The second part follows trivially from \cref{eqn:angle_def} and the choice of \(\alpha\).
\end{proof}

\begin{observation}\label{observation:niertrdlakat}
    For any \(i,c\in\mathbb{N}\) such that \(i+c\leq R\), the following inequality holds:
    \[
        \varphi_{s_i,s_{i+c}}\leq c\alpha\text{.}
    \]
\end{observation}
\begin{proof}
    The inequality holds by the \cref{observation:katalfa} and the triangle inequality for angles.
\end{proof}

If we look at Grover's algorithm we can notice the following facts.

\begin{observation}\label{observation:Grover-rownosc}
    In case of Grover's algorithm we have equality in the inequality given by \cref{observation:niertrdlakat} for \(i + c \leq (\frac{\pi}{\alpha}-1) / 2\).
\end{observation}

\begin{observation}
    In case of Grover's algorithm all points \(s_0,\ldots,s_R\) lie on a great circle of the sphere \(S\).
\end{observation}

\begin{lem}
    For a given \(R\leq(\frac{2\pi}{\alpha}-1)/2\) the expression \(|s_R-s_0|\) is maximised by Grover's algorithm.
\end{lem}
\begin{proof}
    Since in the case of Grover's algorithm points \(s_0,\ldots,s_R\) lie on the great circle we get \(\varphi_{s_0,s_{R}}=R\alpha\) and thus the distance between \(s_0\) and \(s_R\) is maximised.
\end{proof}

Let us here recall the result of Zalka.
By above Lemma, \cref{observation:Grover-rownosc} and \cref{lem:crucial} we get the following.
\begin{corollary}[Zalka's lower bound for search algorithm]\label{cor:Zalka}
    Let \(R\leq(\frac{2\pi}{\alpha}-1)/2\). The Grover's algorithm that makes \(R\) oracle queries gives maximal probability of measuring marked element among all quantum circuits that solve unstructured quantum search problem using at most \(R\) queries.
\end{corollary}
Note that for large \(N\) the number \(\frac{\pi}{2\alpha}\) is close to \(\frac{\pi}{4}\sqrt N\).

Now let us see what happens after two steps of the algorithm.
Put \(d_K=16(K-1)/K\) for \(K\geq1\).

\begin{observation}\label{observation:dwakroki}
    If \(z\in\hat S\), then \(|O_{Id}O_Uz-Uz|^2\leq d_N\).
    In particular for \(i \in \set{0,\ldots,R-2}\) we have
    \[
        |s_i-s_{i+2}|^2\leq d_N\text{.}
    \]
\end{observation}

\begin{proof}
    It is the direct consequence of \cref{observation:niertrdlakat} for \(c=2\).
\end{proof}

The key observation for our lower bound is better estimation for unitary operators that act on bounded number of qubits.
From this point of view we consider that each oracle query can be performed on arbitrary qubits in arbitrary order (or we can think that we just add SWAP gates).
To stress this we use here the symbols \(O\) and \(O'\) for oracle operators.

\begin{lem}\label{lem:technical}
    Let \(z=(\ket{z},\ldots,\ket{z})\in\hat S\). If \(U\) acts at most on \(k\) qubits, then \(|O'_{Id}O_Uz-Uz|^2\leq d_{K^2}\) where \(K=2^k\).
\end{lem}

\begin{proof}
    Let \(A_U\) be the set of \(k\) qubits on which \(U\) acts.
    Oracle \(O\) acts on qubits \(Q_1,\ldots,Q_n\) and \(O'\) on \(Q_1',\ldots,Q_n'\) (here of course the order of qubits is important).
    Let \(J_O\) be a set of all such indices \(i\) that \(Q_i\in A_U\), and \(J_{O'}\) be a set of all incides \(j\) that \(Q_j'\in A_U\).
    Without loss of generality we can assume that \(J=J_O\cup J_{O'} = \{n+1-a,\ldots,n\}\).
    Note that \(a\leq 2k\), since \(|J_O|,|J_O'| \leq |A_U| \leq k\).
    Put
    \[
        B=A_U\cup\{Q_{n+1-a},\ldots,Q_{n}\}\cup\{Q_{n+1-a}',\ldots,Q_{n}'\}\text{.}
    \]
    Let \(B'\) be a set of all other qubits. By the assumption above, it is a prefix of the set of all qubits.

    Let us fix for a moment \(y=(y_1,\ldots,y_n)\in\{0,1\}^{n}\).
    Let \(q=(y_1,\ldots,y_{n-a})\in\{0,1\}^{n-a}\) and \(r=(y_{n-a+1},\ldots,y_{n})\in\{0,1\}^{a}\).
    We will also write \(qr\) in place of \(y\).
    With these notions introduced, we can write \(\ket{z}\) as:
    \[
        \ket{z} =\sum_{\ket{x}\in\mathcal{B}'}\alpha_{x}\ket{x}\ket{S_{x}}\text{,}
    \]
    where \(\mathcal{B'}\) is a computational basis in the space related to qubits from \(B'\), \( \alpha_{x}\) are complex numbers and \(\ket{S_{x}} \) are states in the space related to qubits in \(B\).
    It is clear that
    \[
        \sum_{\ket{x}\in\mathcal{B}'}|\alpha_{x}|^2=1\text{.}
    \]
    We group elements of \(\mathcal{B}'\) into four disjoint sets.
    Let \(\mathcal{B}_1^q\) be the set of all \(\ket x\) that agree with \(y_j\) on qubits \(Q_j\) and \(Q_j'\) for all \(j\leq n-a\).
    Let \(\mathcal{B}_2^q\) (respectively \(\mathcal{B}_3^q\)) be the set of all \(\ket x\) that agree with \(p_j\) on qubits \(Q_j\) (respectively \(Q_j'\)) for all \(j\leq n-a\) but differes on at least one qubit \(Q_j'\) (respectively \(Q_j\)) for some \(j\leq n-a\).
    And finally we put \(\mathcal{B}_4^q=\mathcal{B}'\setminus\left(\mathcal{B}_1^q\cup\mathcal{B}_2^q\cup\mathcal{B}_3^q\right)\).
    Now we have
    \[
        \ket{z}=\sum_{i=1}^4\ket{z_i^q}\text{,}
    \]
    where
    \[
        \ket{z_i^q}=\sum_{\ket{x}\in\mathcal{B}_i^q}\alpha_{x}\ket{x}\ket{S_{x}}\text{.}
    \]
    We have
    \begin{align*}
        U\ket{z_i^q}           & = \sum_{\ket{x}\in\mathcal{B}_i^q}\alpha_{x}\ket{x}U\ket{S_{x}}\text{,}\\
        O'^{y} UO^y\ket{z_1^q} & = \sum_{\ket{x}\in\mathcal{B}_1^q}\alpha_{x}\ket{x}O'^rUO^r\ket{S_{x}}\text{,}\\
        O'^{y} UO^y\ket{z_2^q} & = \sum_{\ket{x}\in\mathcal{B}_2^q}\alpha_{x}\ket{x}UO^r\ket{S_{x}}\text{,}\\
        O'^{y} UO^y\ket{z_3^q} & = \sum_{\ket{x}\in\mathcal{B}_3^q}\alpha_{x}\ket{x}O'^rU\ket{S_{x}}\text{,}\\
        O'^{y} UO^y\ket{z_4^q} & = U\ket{z_4^q}\text{,}
    \end{align*}
    where \(O^r\) (and respectively \(O'^r\)) are oracles on \(a\) qubits that mark element of the computational basis if for \(k>n-a\) on \(Q_{k}\) (respectively on \(Q_{k}'\)) this element is \(y_{k}\).
    We get
    \begin{multline*}
        |O'^{y} UO^{y}\ket{z}-U\ket{z}|^2 \\
        =\sum_{\ket{x}\in\mathcal{B}_1^q}|\alpha_{x}|^2|(U-O'^rUO^r)\ket{S_{x}}|^2 \\
        +\sum_{\ket{x}\in\mathcal{B}_2^q}|\alpha_{x}|^2|(1-O^r)\ket{S_{x}}|^2+\sum_{\ket{x}\in\mathcal{B}_3^q}|\alpha_{x}|^2|(1-O'^r)U\ket{S_{x}}|^2 \text{.}
    \end{multline*}

    Now we are ready to sum up above expression with respect to \(r\).
    For a fixed \(q\), by applying \cref{observation:dwa}, we get
    \begin{multline*}
        \sum_r |O'^{qr} UO^{qr}\ket{z}-U\ket{z}|^2 \\
        =\sum_{\ket{x}\in\mathcal{B}_1^q}|\alpha_{x}|^2\sum_r|(O'^rUO^r-U)\ket{S_x}|^2 \\
        +4\sum_{\ket{x}\in\mathcal{B}_2^q}|\alpha_{x}|^2+4\sum_{\ket{x}\in\mathcal{B}_3^q}|\alpha_{x}|^2 \\
        \leq d_{2^a}\sum_{\ket{x}\in\mathcal{B}_1^q}|\alpha_{x}|^2 + 4\sum_{\ket{x}\in\mathcal{B}_2^q\cup\mathcal{B}_3^q}|\alpha_{x}|^2 \\
        \leq d_{K^2}/2\left(\sum_{\ket{x}\in\mathcal{B}_1^q\cup\mathcal{B}_2^q}|\alpha_{x}|^2+\sum_{\ket{x}\in\mathcal{B}_1^q\cup\mathcal{B}_3^q}|\alpha_{x}|^2\right).
    \end{multline*}
    The inequality in fourth line holds by applying \cref{observation:dwakroki} in case of \(N=2^a\).

    The next step is sum the bound above with respect to \(q\).
    Notice that
    \[
        \bigcup_q \mathcal{B}_1^q\cup\mathcal{B}_2^q =
        \bigcup_q \mathcal{B}_1^q\cup\mathcal{B}_3^q =
        \mathcal{B}'\text{,}
    \]
    since if for a fixed \(q\) one oracle marks some state \(\ket{z}\in\mathcal{B}'\), then the other oracle either agrees with it (putting \(\ket{z}\) in \(\mathcal{B}_1^q\)) or not (putting it in \(\mathcal{B}_2^q\) or \(\mathcal{B}_3^q\), respectively).
    Because of that and the fact that for different \(q\)s oracle marks disjoint \(\ket{z}\)s, the \(\cup\) symbol is to be understood as disjoint set union.
    Therefore we have
    \[
        \ket{z}=\sum_q\sum_{\ket{x}\in\mathcal{B}_1^q\cup\mathcal{B}_2^q}\alpha_{x}\ket{x}\ket{S_{x}} =\sum_q\sum_{\ket{x}\in\mathcal{B}_1^q\cup\mathcal{B}_3^q}\alpha_{x}\ket{x}\ket{S_{x}}
    \]
    and thus
    \[
        \sum_q\sum_{\ket{x}\in\mathcal{B}_1^q\cup\mathcal{B}_2^q}|\alpha_{x}|^2 =\sum_q\sum_{\ket{x}\in\mathcal{B}_1^q\cup\mathcal{B}_3^q}|\alpha_{x}|^2 =1\text{.}
    \]

    Finally we can conclude
    \begin{multline*}
        |O'_{Id}O_Uz-Uz|^2=\sum_y |O'^{y} UO^y\ket{z}-U\ket{z}|^2 \\
        =\sum_q\sum_r |O'^{qr} UO^{qr}\ket{z}-U\ket{z}|^2\leq d_{K^2}\text{.}
    \end{multline*}
\end{proof}

\begin{proof}[\bf Proof of \cref{thm:lowerbound}]
    Let us choose \(k=8\eta\).
    Note that if \(k>n/4\) then for any \(C\geq32\) we get \(T2^{-C\eta}\leq (\pi\sqrt N/4+1)/N<1\) and we are done, so we can assume that \(k\leq n/4\).

    By \cref{observation:niertrdlakat} for all \(i \in \set{0,\dots,R-2}\) we have
    \[
        \varphi_{s_i,s_{i+2}}\leq 2\alpha=2\arcsin\left(\sqrt{d_{N}}/2\sqrt{N}\right)\text{,}
    \]
    and by \cref{lem:technical}, if operator \(U_{i+1}\) acts on at most \(k\) qubits then
    \[
        \varphi_{s_i, s_{i + 2}} \leq 2 \arcsin \paren*{ \frac{\sqrt{d_{K^2}}}{2 \sqrt{N}} }\text{.}
    \]
    We will use either of these bounds depending on whether an operator acts on more than \(k\) qubits or not.
    Note that the second bound is always better than the first one, as we assumed that \(k \leq n / 4\).

    Since \(\arcsin\) has the derivative greater or equal to one, for \(U_{i+1}\) acting on at most \(k\) qubits we can bound
    \begin{align*}
        \varphi_{s_i,s_{i+2}} &\leq 2 \arcsin \paren*{ \frac{\sqrt{d_N}}{2 \sqrt{N}} } - 2 \frac{\sqrt{d_N} - \sqrt{d_{K^2}}}{2 \sqrt{N}} \\
        &\leq 2 \alpha - \frac{D}{K^2 \sqrt{N}}\text{,}
    \end{align*}
    where the constant \(D\) (as well as constants \(D'\) and \(C\) below) does not depend on \(N\) and \(K\).
    The inequality \(\arcsin x\leq 2x\) for \(x=1/\sqrt{N}\) combined with \cref{eqn:defalpha} yields
    \[
        \varphi_{s_i,s_{i+2}}\leq 2\alpha(1-D/8K^2)\text{.}
    \]
    From the triangle inequality for angles we can now establish the following bound
    \[
        \varphi_{s_R,s_0} \leq\hat\alpha+\sum_{t\in\{l\in\mathbb{N}:2l+2\leq R\}}\varphi_{s_{2t},s_{2t+2}}\text{,}
    \]
    where
    \[
        \hat\alpha=
            \begin{cases}
                \alpha & \text{ if \(R\) is an odd number,} \\
                0      & \text{ if \(R\) is an even number.}
            \end{cases}
    \]
    Note that, since each basic gate acts on at most two qubits, at least half of operators \(U_{2t+1}\) act on at most \(k\) qubits.
    Therefore we can bound half of the angles by \(2\alpha(1-D/8K^2)\) and the rest by \(2\alpha\), which gives us
    \[
        \varphi_{s_R,s_0} \leq R\alpha-\frac{D'R}{K^2}\alpha\leq\alpha R(1-2^{-C\eta})\text{.}
    \]
    On the other hand by \cref{lem:crucial} and \cref{observation:Grover-rownosc} we have
    \[
        \varphi_{s_R,s_0} >(T-1)\alpha\text{,}
    \]
    and thus
    \[
        R> (T-1)/(1-2^{-C\eta})\geq T-1+2^{-C\eta}T
    \]
    and the Theorem follows.
\end{proof}

\begin{proof}[proof of \cref{corollary:optimality}]
    One can see that any algorithm needs more than \(\frac{\pi}{4}\sqrt{N}-1\) steps to find the marked element with certainty (compare \cref{cor:Zalka}).
    Using \cref{thm:lowerbound} we get the result for \(\delta=1/C\) and for large enough \(n\).
    If necessary, we may decrease \(\delta\) so that \( \delta\log\paren*{\sfrac{1}{\varepsilon}} \sqrt{N}<1\) for smaller values of \(n\).
\end{proof}

\begin{remark}
    Note that we do not allow measurements before the end of the algorithm.
    It is not clear to the authors how measurements performed inside a circuit can reduce the expected number of oracle queries made.
    In particular, how Zalka’s results about optimality of Grover’s algorithm applies to this more general class of quantum algorithms is far from obvious.
    Example of measurements speeding up quantum procedures can be found in section 4 of \cite{inbook} or in the last section of this paper.
\end{remark}

\section{Partial uncompute}
\label{section:uncompute}

\subsection{Motivation and intuition}

The motivation for this section comes from the fact that many natural implementations of phase oracles mimic parallel
classical computation by the following pattern of operations.
\begin{enumerate}
    \item We perform a long series of operations that do not alter the original $n$ qubits (or alter them temporarily), but modify some number of ancilla qubits that were initially zero, usually by CCX gates.
    \item We perform a Z gate operation to flip the phase on the interesting states.
    \item We undo all the operations from step (1), not to hinder the amplitude interference in the
    subsequent mixing operators.
\end{enumerate}
Step (3) offers the benefit of being able to reuse the (now cleared) ancilla qubits, but it is not the
main motivation of performing it.

If a mixing operators used only acts on a subset of qubits, maybe not all gates from step (1) interfere
with proper amplitude interference? It turns out that very often it suffices to undo only a fraction of
gates. We will establish the proper language to express that in \cref{ss:uncompute-definitions}.

It turns out that the state that allows safe application of the mixing operator is much closer 
(in the metric of number of gates) to the state from after step (1) that to the base state with all
ancilla qubits zeroed. Our approach will initially compute all the ancilla qubits, perform all the
mixing ``close'' to this state, and finally uncompute the ancillas.

Intuitively, we shall follow the following new scheme.
\begin{enumerate}
    \item Compute all the ancilla qubits.
    \item For all mixing operators, perform the following.
    \begin{enumerate}
        \item Perform the phase flip.
        \item Undo the ancilla computation that would interfere with the upcoming mixing operator.
        \item Perform the mixing.
        \item Redo the computation from step 2(b).
    \end{enumerate}
    \item Uncompute all the ancilla qubits.
\end{enumerate}
The last step (2d) and step (3) could be even skipped, if the ancilla computation does not modify the 
original qubits. However we are not doing this optimization in the formal approach, as the benefits
are minimal.

Naturally, this is a very imprecise description. Full details are presented in  \cref{ss:uncompute-reducing}.

We are aware that many (if not all) of these operations are performed by modern quantum circuit
optimizers and preprocessors. The aim is to give structure to the process and understand how many
gates are guaranteed to be removed from the circuit.

\subsection{Definitions}
\label{ss:uncompute-definitions}

Recall from \cref{section:premilinaries}, that a phase oracle \(O\) of a function \(f \colon \set{0,1}^n \rightarrow \set{0, 1}\) is a unitary transformation given by \(O \ket{x} = (-1)^{f(x)} \ket{x}\) for all vectors \(\ket{x}\) in the computational basis.
\begin{ddef}\label{def:uncomputable}
We will say that a phase oracle $O$ admits an \textit{uncomputable decomposition} $(O_u,O_p)$ if
$O=O_u^{\dagger}\circ O_p\circ O_u$. We call $O_u$ and $O_p$ the uncomputable part and the phase part,
respectively.
\end{ddef}

\begin{remark} Note that neither \(O_u\) nor \(O_p\) need to be phase oracles themselves. Naturally, every phase oracle $O$ on $n$ qubits admits a trivial decomposition $(\id_n,O)$.
However, in practice, many real-life examples give more interesting decompositions. The intuitive goal is to make the phase part as simple as possible. The more gates needed to implement \(O\) we manage to move to the uncomputable part, the more gates we can hope to cancel out.
A comon pattern in many settings is to use an ancilla to mark the sought state by a bitflip, apply Z gate on said ancilla, followed by uncomputing the ancilla. In such case, there is a natural uncomputable decomposition of \(O\).
\end{remark}

\begin{ddef}\label{def:depends}
Let $A_1,\dots,A_\ell$ be a chronologically ordered sequence of unitary matrices corresponding to the 
gates in a quantum circuit operating on $n$ qubits (ties being broken arbitrarily). We will define the
fact that \textit{$A_j$ depends on qubit $q_i$}, $i\in\{1,\dots,n\}$, by 
induction on $j\in\{1,\dots,\ell\}$.

We say that $A_j$ depends on $q_i$ if $A_j$ acts on $q_i$ or if there exist $i_0\in\{1,\dots,n\}$ and 
$j_0\in\{1,\dots,j-1\}$ such that the following hold.
\begin{itemize}
    \item $A_{j_0}$ depends on $q_i$.
    \item $A_{j_0}$ acts on $q_{i_0}$.
    \item $A_j$ acts on $q_{i_0}$.
\end{itemize}
\end{ddef}

In this section, by \(\calP(X)\) we will denote the \textit{powerset of } \(X\), that is the set of all subsets of \(X\).
\begin{ddef}
Let $O$ be a phase oracle on $n$ qubits, $\ell\in\mathbb N$,
$d:\{1,\dots,\ell\}\longrightarrow\mathcal P(\{q_1,\dots,q_n\})$,
and $U:\{1,\dots,\ell\}\ni j\longmapsto U_j\in U(n)$.
Assume that $U_j$ is an arbitrary unitary operator acting on the qubit set $d(j)$,
$j\in\{1,\dots,\ell\}$. We
define a \textit{generic oracle circuit} $V(\ell,d,U,O)$ by the following 
formula (the product is to be understood as right--to--left operator composition):
$$ 
V(\ell,d,U,O):=\prod_{j=1}^\ell
\left(U_j\circ O \right)
.
$$
\end{ddef}

\begin{remark}
Observe that $W_m$ from \cref{def:Wojter} is a generic oracle circuit.
\end{remark}

\begin{proof}
For $j\in\{1,\dots,m\}$ put $\ell_j:=(3^j-1)/2$ and define 
$d_j:\{1,\dots,\ell_j\}\longmapsto\mathcal P(\{1,\dots,k_1+\dots+k_j\})$ recursively
as follows:
\begin{widetext}
$$
d_j(i):=
\begin{cases}
d_{j-1}(i)
,&
\quad
1\leq i\leq\ell_{j-1}
,
\\
\{k_1+\dots+k_{j-1}+1,\dots,k_1+\dots+k_j\}
,&
\quad
i=\ell_{j-1}+1
,
\\
d_{j-1}(2\ell_{j-1}+2-i)
,&
\quad
\ell_{j-1}+2\leq i\leq2\ell_{j-1}+1
,
\\
d_{j-1}(i-2\ell_{j-1}-1)
,&
\quad
2\ell_{j-1}+2\leq i\leq\ell_j
.
\end{cases}
$$
\end{widetext}
We then set $U_j$ to be the mixing operator $G_{|d_m(j)|}$ applied onto the qubits in the set 
$d_m(j)$, $j\in\ell_m$. Observe that $W_m=V(\ell_m,d_m,U,O)$.
\end{proof}

\subsection{Reducing the number of gates}
\label{ss:uncompute-reducing}

\begin{thm}\label{thm:subsets}
Let $(O_u,O_p)$ be an uncomputable decomposition of $O$ and let $D_u$, and $D_p$, be the total number 
of gates used in $O_u$, and $O_p$, respectively. Let $\overline D_s$ denote the number of gates
within $O_u$ that depend on any of the qubits in $s$, for all $s\in\mathcal P(\{1,\dots,n\})$.

For a given generic oracle circuit $V(\ell,d,U,O)$ one can implement an equivalent circuit
$\widetilde V$ that uses a total of $2D_u+\ell D_p+2\sum_{j=1}^\ell\overline D_{d(j)}$ gates for oracle
queries. This results in no more than the following average number of gates per oracle query:
$$
D_p+2\frac{D_u}{\ell}+2\frac{\sum_{j=1}^\ell\overline D_{d(j)}}{\ell}.
$$
\end{thm}

\begin{proof}
Let $O_s$ (resp. $\widetilde O_s$) be the in-order composition of all gates in $O_u$ that depend
(resp. do not depend) on any of the qubits in $s$, for all $s\in\mathcal P(\{1,\dots,n\})$.

First, let us observe that, for a fixed $s\in\mathcal P(\{1,\dots,n\})$, every gate in $A\in O_s$
commutes with every
gate in $\widetilde O_s$ that originally appears later than $A$, as there are no common qubits that they 
act on. This implies that $O_u=O_s\circ\widetilde O_s$, for all $s\in\mathcal P(\{1,\dots,n\})$.

Similarly $U_j$ and $\widetilde O_{d(j)}$ commute,
as there are no common qubits that they act on, $j=1,\dots,\ell$.

We now proceed to apply these properties to $V(\ell,d,U,O)$ in order to obtain $\widetilde V$.
This will happen in the following five steps.
\begin{enumerate}
    \item Append an identity $O_u^\dagger\circ O_u$ operation to each factor of $V(\ell,d,U,O)$, see \cref{fig:uncompute-step-1}.
    \item Express $O_u$ (resp. $O_u^\dagger$) occuring next to $U_j$ as $O_{d(j)}\circ\widetilde O_{d(j)}$ (resp. $\widetilde O_{d(j)}^\dagger\circ O_{d(j)}^\dagger$ in the $j$th factor, $j\in\{1,\dots,\ell\}$, see \cref{fig:uncompute-step-2}.
    \item Swap $\widetilde O_{d(j)}$ and $U_j$ in the $j$th factor, $j\in\{1,\dots,\ell\}$, see \cref{fig:uncompute-step-3}.
    \item Remove the identity operation $\widetilde O_{d(j)}^\dagger\circ\widetilde O_{d(j)}$ in the $j$th factor, $j\in\{1,\dots,\ell\}$, see \cref{fig:uncompute-step-4}.
    \item Remove the identity operation $O_u\circ O_u^\dagger$ on the boundary between each two consecutive factors, see \cref{fig:uncompute-step-5}.
\end{enumerate}

    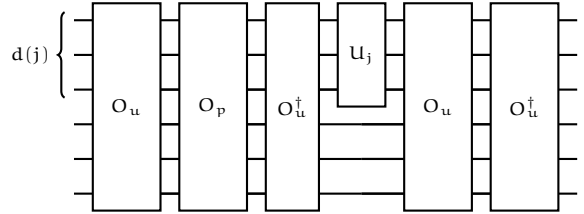
\begin{figure}[H]
            {\scriptsize
\begin{quantikz}[row sep=0.01cm, column sep=0.25cm]
        \lstick[3]{$d(j)$}     & \gate[wires=6,disable auto height][0.9cm]{O_u} & \gate[6,disable auto height][0.9cm]{O_p}  & \gate[6,disable auto height]{O_u^\dagger} & \gate[3,disable auto height]{U_j}           & \gate[6,disable auto height][0.9cm]{O_u} & \gate[6,disable auto height][0.9cm]{O_u^\dagger} &\qw \\
            \lstick{}          & \ghost{W}                      & \qw                  & \qw                                & \qw                      & \qw&\qw&\qw \\
            \lstick{}          & \ghost{W}                      & \qw                  & \qw                                & \qw                      & \qw&\qw&\qw \\ 
            & \ghost{W}                      & \qw                  &  \qw & \qw                      & \qw&\qw&\qw \\
            \lstick{}          & \ghost{W}                      & \qw                  & \qw                                & \qw                      & \qw&\qw&\qw \\ 
            \lstick{}          & \ghost{W}                      & \qw                  & \qw                                & \qw                      & \qw&\qw&\qw \\
        \end{quantikz}}
        \caption{Graphical representation of the $j$th factor after applying step (1).
        For simplicity we assumed that $d(j)$ consists of first $|d(j)|$ qubits.}
        \label{fig:uncompute-step-1}
    \end{figure}
    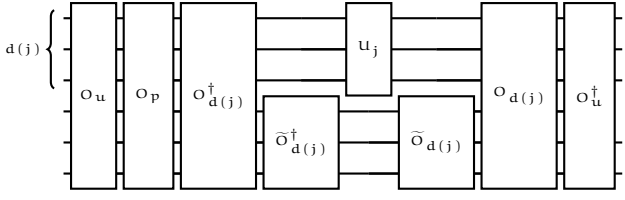
\begin{figure}[H]
            {\tiny
\begin{quantikz}[row sep=0.01cm, column sep=0.1cm]
           \lstick[3]{$d(j)$}     & \gate[6,disable auto height][0.6cm]{O_u} & \gate[6,disable auto height][0.6cm]{O_p} & \gate[6,disable auto height][0.6cm]{O_{d(j)}^\dagger} & \qw & \gate[3,disable auto height]{U_j}          & \qw & \gate[6,disable auto height][0.7cm]{ O_{d(j)}}   & \gate[wires=6,disable auto height][0.6cm]{O_u^\dagger} &\qw \\
            \lstick{}          & \ghost{W}                      & \qw                  & \qw             &\qw                   & \qw                      &\qw& \qw&\qw&\qw \\
            \lstick{}          & \ghost{W}                &\qw      & \qw                  & \qw             &\qw                   & \qw                      & \qw&\qw&\qw \\ 
            & \ghost{W}        & \qw & \qw & \gate[3,disable auto height][0.6cm]{\widetilde O^\dagger_{d(j)}}      & \qw& \gate[3,disable auto height][0.6cm]{\widetilde O_{d(j)}}& \qw &\qw \\
            \lstick{}          & \ghost{W}        &\qw              & \qw                  & \qw             &\qw                   & \qw                      & \qw&\qw&\qw \\ 
            \lstick{}          & \ghost{W}            &\qw          & \qw                  & \qw             &\qw                   & \qw                      & \qw&\qw&\qw \\
        \end{quantikz}}
        \caption{Graphical representation of the $j$th factor after applying step (2).
        For simplicity we assumed that $d(j)$ consists of first $|d(j)|$ qubits.}
        \label{fig:uncompute-step-2}
    \end{figure}

    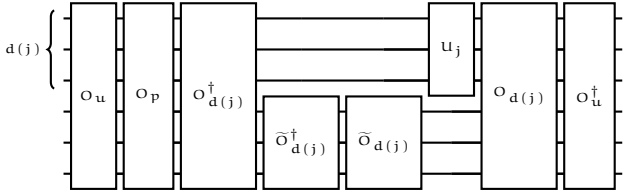
\begin{figure}[H]
            {\tiny
\begin{quantikz}[row sep=0.01cm, column sep=0.1cm]
           \lstick[3]{$d(j)$}  & \gate[6,disable auto height][0.6cm]{O_u} & \gate[6,disable auto height][0.6cm]{O_p} & \gate[6,disable auto height][0.6cm]{O_{d(j)}^\dagger} &\qw         &\qw &  \gate[3,disable auto height]{U_j} & \gate[6,disable auto height][0.6cm]{ O_{d(j)}}   & \gate[wires=6,disable auto height][0.6cm]{O_u^\dagger} &\qw \\
            \lstick{}          & \ghost{W}                      & \qw                  & \qw             &\qw                   & \qw                      &\qw& \qw&\qw&\qw \\
            \lstick{}          & \ghost{W}                &\qw      & \qw                  & \qw             &\qw                   & \qw                      & \qw&\qw&\qw \\ 
            & \ghost{W}                      &  \qw           &\qw &\gate[3,disable auto height][0.6cm]{\widetilde O^\dagger_{d(j)}}      &  \gate[3,disable auto height][0.6cm]{\widetilde O_{d(j)}}& \qw                      & \qw&\qw&\qw \\
            \lstick{}          & \ghost{W}        &\qw              & \qw                  & \qw             &\qw                   & \qw                      & \qw&\qw&\qw \\ 
            \lstick{}          & \ghost{W}            &\qw          & \qw                  & \qw             &\qw                   & \qw                      & \qw&\qw&\qw \\
        \end{quantikz}}
        \caption{Graphical representation of the $j$th factor after applying step (3).
        For simplicity we assumed that $d(j)$ consists of first $|d(j)|$ qubits.}
        \label{fig:uncompute-step-3}
    \end{figure}

    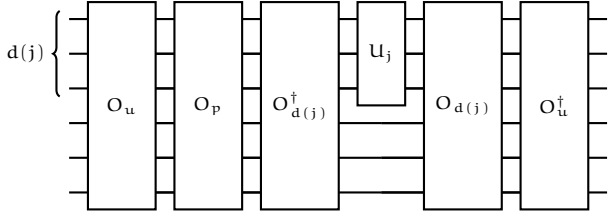
\begin{figure}[H]
            {\scriptsize
\begin{quantikz}[row sep=0.01cm, column sep=0.25cm]
           \lstick[3]{$d(j)$}  & \gate[6,disable auto height][0.9cm]{O_u} & \gate[6,disable auto height][0.9cm]{O_p} & \gate[6,disable auto height][0.9cm]{O_{d(j)}^\dagger} &  \gate[3,disable auto height]{U_j} & \gate[6,disable auto height][0.9cm]{ O_{d(j)}}   & \gate[wires=6,disable auto height][0.9cm]{O_u^\dagger} &\qw \\
            \lstick{}          & \ghost{W}                      & \qw                  & \qw             &\qw              & \qw&\qw&\qw \\
            \lstick{}          & \ghost{W}                &\qw      & \qw                  & \qw             &\qw           &\qw&\qw \\ 
            & \ghost{W}                      &  \qw           &\qw & \qw                      & \qw&\qw&\qw \\
            \lstick{}          & \ghost{W}        &\qw              & \qw                  & \qw                                   & \qw&\qw&\qw \\ 
            \lstick{}          & \ghost{W}            &\qw          & \qw                  & \qw                                   & \qw&\qw&\qw \\
        \end{quantikz}}
        \caption{Graphical representation of the $j$th factor after applying step (4).
        For simplicity we assumed that $d(j)$ consists of first $|d(j)|$ qubits.}
        \label{fig:uncompute-step-4}
    \end{figure}

    \begin{figure}[H]
    {\scriptsize
        \begin{quantikz}[row sep=0.01cm, column sep=0.25cm]
           \lstick[3]{$d(j)$}  & \gate[6,disable auto height][0.9cm]{O_p} & \gate[6,disable auto height][0.9cm]{O_{d(j)}^\dagger} &  \gate[3,disable auto height]{U_j} & \gate[6,disable auto height][0.9cm]{ O_{d(j)}}   & \qw \\
            \lstick{}          & \ghost{W}                      & \qw                  & \qw             &\qw              & \qw\\
            \lstick{}          & \ghost{W}                &\qw      & \qw                  & \qw             &\qw            \\ 
            & \ghost{W}                      &  \qw           &\qw & \qw                      & \qw \\
            \lstick{}          & \ghost{W}        &\qw              & \qw                  & \qw                                   & \qw \\ 
            \lstick{}          & \ghost{W}            &\qw          & \qw                  & \qw                                   & \qw \\
        \end{quantikz}
        }
        \caption{Graphical representation of the $j$th factor after applying step (5).
        For simplicity we assumed that $d(j)$ consists of first $|d(j)|$ qubits.
        Notice that this picture is only valid for $1<j<\ell$ as the first factor will retain $O^\dagger_u$, while the last one will retain $O_u$.}
        \label{fig:uncompute-step-5}
    \end{figure}
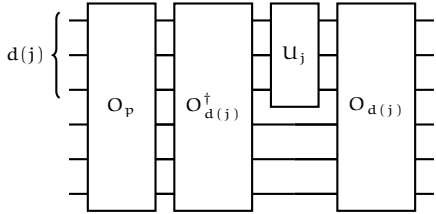

More concisely we put:
$$
\widetilde V:=O_u^\dagger
\circ
\prod_{j=1}^\ell
\left(
O_{d(j)}
\circ
U_j
\circ
O_{d(j)}^{\dagger}
\circ
O_p
\right)
\circ
O_u
.
$$
As discussed above $V(\ell,d,U,O)$ and $\widetilde V$ are equal as 
unitary operators. The desired gate count follows directly from the definition.
\end{proof}

\begin{corollary}\label{cor:weighted_average}
Let $(O_u,O_p)$ be an uncomputable decomposition of $O$ and let $D_u$, and $D_p$, be the total number 
of gates used in $O_u$, and $O_p$, respectively. 
Let $D_i$ be the number of gates within $O_u$ that depend on the $i$th qubit, $i\in\{1,\dots,n\}$ and
let $\overline D_j$ be the average of $D_i$ taken over the qubits $i\in d(j)$, i.e., the average number of gates
within $O_u$ that depend on a single qubit from the $d(j)$, $j=1,\dots,\ell$.

For a given generic oracle circuit $V(\ell,d,U,O)$ one can implement an equivalent circuit
$\widetilde V$
that uses no more than the following average number of gates per oracle query:
$$
D_p+2\frac{D_u+\sum_{j=1}^\ell|d(j)|\overline D_{j}}{\ell}.
$$
\end{corollary}

\begin{proof}
This follows directly from \cref{thm:subsets} and the fact that the number of gates that depend on
any of the qubits from a given set is no greater that the sum of the numbers of gates that depend
on individual qubits from this set.
\end{proof}

\begin{corollary}\label{cor:average}
Let $(O_u,O_p)$ be an uncomputable decomposition of $O$ and let $D_u$, and $D_p$, be the total number 
of gates used in $O_u$, and $O_p$, respectively. 
Let $D_i$ be the number of gates within $O_u$ that depend on the $i$th qubit, $i\in\{1,\dots,n\}$ and
let $D$ be the average of $D_i$ taken over all qubits $i\in\{1,\dots,n\}$.

For a given generic oracle circuit $V(\ell,d,U,O)$ one can implement an equivalent circuit
$\widetilde V$
that uses no more than the following average number of gates per oracle query:
$$
D_p+\frac{2D_u}{\ell}+2D\frac{\sum_{j=1}^\ell|d(j)|}{\ell}.
$$

In particular, one can implement a circuit equivalent to $W_m$ that uses an average of no more than
the following number of gates per oracle query:
$$
D_p+\frac{4D_u}{3^m-1}+4D\frac{\sum_{j=1}^mk_j3^{m-j}}{3^m-1}.
$$
Using the notation of \cref{thm:main}, this asymptotic average number of gates is
$\mathcal O(D_p+\log(1/\varepsilon)D)$.
\end{corollary}

\begin{proof}
Without loss of generality, we may assume that $(D_i)_{i=1}^n$ is non--decreasing. Similarly, without
loss of generality, we may assume that the weights $w_i = |\{j\in\{1,\dots,\ell\}:i\in d(j)\}|/\ell$, $i=1,\dots,n$
of subsequent qubits form a non--increasing sequence.
Then a weighted average of $D_i$s:
$$
  \frac{\sum_{i=1}^n w_i D_i}{\sum_{i=1}^n w_i}
  = \frac{\sum_{j=1}^\ell\sum_{i\in d(j)} D_i}{\ell\sum_{i=1}^n w_i}
  = \frac{\sum_{j=1}^\ell |d(j)|\overline D_j}{\sum_{j=1}^\ell d(j)}
$$
will not be greater than their arithmetic mean $D$. This completes the first part of the proof.
To obtain bound for circuit $W_m$, recall from \cref{section:circuit} that $\ell = (3^m - 1) / 2$
and that $j$-th diffuser (of size $k_j$) appears in $W_m$ exactly $3^{m-j}$ times, so
$\sum_{j=1}^\ell |d(j)| = \sum_{j=1}^m k_j 3^{m-j}$.

For the asymptotic result, recall from the proof of \cref{thm:main} (\cref{eqn:depthbound}) that
(after setting $x\in\Theta\paren*{\log(1/\varepsilon)}$ as in the proof) we have
$m=\Theta \paren*{\sqrt{n/\log(1/\varepsilon)}}$ and we invoke the $W_m$ circuit $\mathcal O(2^{n/2}/3^m)$ times
with parameters $k_j\leq3\ceil{\log_2(1/\varepsilon)}j$, $j\in\{1,\dots,m\}$.
Additionally, observe that
\begin{multline*}
\sum_{j=1}^mj3^{m-j}=\sum_{j=1}^m\frac{3^{m-j+1}-1}2
\\
=\frac34(3^m-1)-\frac m2=\mathcal O(3^m-1).
\end{multline*}
As each amplitude amplification step requires additional $\mathcal O(n+D_p+D_u)$ gates, we get that the asymptotic 
average number of gates is
\begin{multline*}
D_p+\mathcal O(D_u/3^m)+\log(1/\varepsilon)\mathcal O(D)+\mathcal O((n+D_p+D_u)/3^m)
\\
=\mathcal O(D_p+\log(1/\varepsilon)D).
\end{multline*}
\end{proof}

\begin{corollary}\label{cor:ksat}
Let $K\in\mathbb N$ and consider all Unique $k$-SAT instances, each with $n$ variables and $c$ clauses. 
For each such an instance there exists a circuit equivalent to $W_m$ that uses an average of no more than 
the following number of gates per query:
\begin{multline*}
1+\frac{12Kc+8c-4}{3^m-1}
\\
+
\frac{4Kc(4+\ceil{\log_2K}+\ceil{\log_2c})}n\cdot\frac{\sum_{j=1}^mk_j3^{m-j}}{3^m-1}.
\end{multline*}
Using the notation of \cref{thm:main}, the asymptotic average per--oracle--query number of gates of a 
quantum circuit solving Unique $k$-SAT with certainty is
$\mathcal O(\log(1/\varepsilon)c\log(c)/n)$.
\end{corollary}

\begin{proof}
A straightforward implementation 
of the phase oracle $O$ consists of $D_u\leq 3Kc+2c-1$ gates (each being an X, a CX, or a CCX) in the 
uncomputable part and one Z gate ($D_p=1$) in the phase part.

More precisely, we introduce the following ancilla qubits and the gates to compute them:
\begin{enumerate}
    \item $c$ qubit groups of $K$ qubits corresponding to negation of all clause literals, each computed by one CX and at most one X,
    \item $c$ qubit groups of $K-1$ qubits corresponding to conjunctions of qubits from (1), each computed by one CCX. They are arranged into a binary tree, so that only $\ceil{\log_2K}$ gates depend on every qubit of (1).
    \item $c$ qubits corresponding to all clauses, each computed by one CX and one X from a top--level qubit of (2).
    \item $c-1$ qubits corresponding to conjunctions of qubits from (3), each computed by one CCX. Again, they are arranged into a binary tree, so that only $\ceil{\log_2c}$ gates depend on every qubit from (3).
\end{enumerate}

For a variable $v$ appearing in $c_v$ clauses, the 
number of gates depending on $v$ is at most $c_v(2+\ceil{\log_2K}+2+\ceil{\log_2c})$, so we get
the average $D\leq Kc(4+\ceil{\log_2K}+\ceil{\log_2c})/n$. By Corollary \ref{cor:average} we get both 
claims.
\end{proof}

\ksat*

\begin{proof}
This follows directly from \cref{thm:main} and \cref{cor:ksat}.
\end{proof}

\section{Multipoint Oracle}
\label{section:multipoint}
    We now proceed to the unstructured search problem with multiple marked elements. As in previous sections we  assume that number of qubits in the input of the oracle is \( n \).
    Let \(S\) be the set of elements marked by oracle \(O\) and let \(K=|S|>0\). For convenience we  mostly refer to the number  \(k = 1+\lceil \log_2{K} \rceil \). We begin our investigation with the assumption that \(k\) is known in advance and later proceed to consider the harder case of unknown \(k\).
    \subsection{Known Number of Marked Elements}
    In this section we assume that value \(k\) is known. It is weaker assumption than knowing \(K\) but  it is sufficient for our purposes.

    We use algorithm from \Cref{thm:main} as a subroutine in algorithms in this section.
    By \emph{\(\texttt{SinglePoint}(O,n)\)} we denote the algorithm from \Cref{thm:main} that solves unstructured search problem for oracle \(O\) which marks exactly one element and acts on \(n\) qubits.
    We want to reduce the problem of unstructured search with possibly many elements marked to the unstructured search with one marked element. To do this we construct a family of hash functions that allows us to effectively parametrize a subset of \( \{0,1\}^n \) which with high probability contains only one element from \(S\). This technique is nearly identical to reduction from SAT to Unique SAT presented in \cite{valiant}. Next, we improve some aspects  of this reduction, so that methods of partial uncompute may be used to reduce the number of additional non-oracle basic gates.

    Let us recall that family \( U \) of hash functions from \(X \) to \( Y \), both being finite is called \emph{pairwise independent} if for every \(x \in X\) and every \(y \in Y\) we have:
    \[\PP_{h \in_R U}\paren*{h(x)=y} = \frac{1}{|Y|} \]
    and for every \(x_1, x_2  \in X\), \(x_1 \neq x_2\) and every \( y_1, y_2 \in Y \) we have:

    \[ \PP_{h \in_R U}\paren*{h(x_1)=y_1 \wedge h(x_2) = y_2} = \frac{1}{|Y|^2}. \]

    We use the following formulation of results from \cite{valiant} that can be found in \cite[p9.10 (180)]{arora}.

    \begin{lem}[Valiant-Vazirani]\label{valian-vazirani}
    For any family of pairwise independent hash functions \(H
    \) from \(\{0,1\}^{n}\) to \(\{0,1\}^{k}\) and \(S \subset \{0,1\}^{n} \) such that \(2^{k-2}\leq |S| \leq 2^{k-1}\) we have that
    \[ \PP_{h \in_R H}\paren*{ |\{x \in S : h(x)=0\}| =1 } \geq \frac{1}{8}\]
    \end{lem}

    A function \(h:\{0,1\}^m \rightarrow \{0,1\}^l\) is called \emph{affine} if we may represent it as \(h(x)=Ax+b\) for some \(A \in \{0,1\}^{l\times m} \) and \( b \in \{0,1\}^{l}\). All arithmetical operations are performed modulo \(2\).

    The \emph{kernel} of the affine function \(h:\{0,1\}^m \rightarrow \{0,1\}^l\) of the form \(h(x)=Ax+b\) is defined as \( \ker{h} = h^{-1}(0) \). Whenever the kernel of the affine function \(h\) is not an empty set, we define the dimension of this kernel as \( \dim{\ker{h}} = \dim{\ker{A}} \). Where \( \ker{A} \) is the null space of the matrix \(A\).

    Whenever function \(h\) is clear from the context we will use \(d=\dim{\ker{h}} \) for brevity.
    Our choice of the family of hash functions is as follows.

    \begin{ddef}
    Define a family of hash functions \( H_{n,k}\) as the set of all affine maps from \(\{0,1\}^{n}\) to \(\{0,1\}^{k}\):
    \begin{multline*}
    H_{n,k} \\
    := \{ h_{A,b} \colon A \in \{0,1\}^{k\times n}, b\in \{0,1\}^k, h_{A,b}(x)=Ax+b\}.
    \end{multline*}
    \end{ddef}

    The first mention of this family is in \cite{carter}, more detailed considerations can be found in \cite{toeplitz}.
    The following standard result will be of use to us.
    \begin{observation}[Folklore]\label{independence}
    The family \(H_{n,k}\) is pairwise independent.
    \end{observation}

    We would like to run algorithm \(\texttt{SinglePoint}\) on the set \(\ker{h} \). To do so we parametrize \(\ker{h}\) by some injection \( g: \{0,1\}^{\dim{\ker{h}}} \rightarrow \ker{h} \) and build a quantum oracle \(O_g\) defined as follows.

    \begin{ddef}\label{oprim}
    Given a quantum oracle \(O : (\mathbb{C}^2)^{\otimes n} \rightarrow (\mathbb{C}^2)^{\otimes n} \), and any function  \( g: \{0,1\}^{d} \rightarrow \{0,1\}^n \) we define the \(g\)-\emph{restricted oracle}
    \( O_g \) as:
    \[O_g=D_g^{-1}(Id_d\otimes O) D_g, \]
    where \(D_g\) is a unitary operator on \((\mathbb{C}^2)^{\otimes (d+n)} \) whose action on states \( \ket{i}\ket{0 \ldots 0} \) for \(i \in \{0,1\}^{d}\) is defined as
    \[ D_g\ket{i}\ket{0 \ldots 0} = \ket{i}\ket{g(i)}. \]
    \end{ddef}

    \begin{observation}
    If oracle \(O\) admits uncomputable decomposition \((O_u, O_p)\) then for any function \(g : \{0,1\}^d \rightarrow \{0,1\}^n \) the \(g\)-restricted oracle \(O_g\) admits an uncomputable decomposition \( ((Id_d\otimes O_u) D_g, Id_d\otimes O_p) \).
    \end{observation}
    \begin{proof}
    It follows directly from definition
    \begin{multline*}
    O_g = D_g^{-1}(Id_d\otimes O) D_g
    \\
    =D_g^{-1}(Id_d\otimes O_{u}^{-1})(Id_d\otimes O_{p})(Id_d\otimes O_{u}) D_g
    \\
    = ((Id_d\otimes O_{u}) D_g)^{-1} (Id_d\otimes O_{p}) ((Id_d\otimes O_{u}) D_g).
    \end{multline*}

    \end{proof}

    \begin{lem}\label{dekorator}
    For an affine injective function \(g:\{0,1\}^{d} \rightarrow \ker{h} \) of the form \(g(x)=Cx+p\), where \(C=(c_{ij})\) we can construct \(D_g\) using basic quantum gates so that the number of gates that depend on \(j\)-th qubit of the first register is exactly equal to \(|\{i \colon c_{ij} = 1\}|\).
    \end{lem}
    \begin{proof}
    It is easy to see that \(D_g\) can be implemented using gates \(CX(f_j, s_i)\) where \(f_j\) is the \(j\)-th qubit of the first register and \(s_i\) is  the \(i\)-th qubit of the second register for each \(i,j\) such that \(c_{ij}=1\) and using gates \(X(s_i)\) for all \(i\) such that \(p_i>0\).

    \end{proof}

    Now we are ready to construct \(g\) which effectively parametrizes \(\ker{h}\).

    \begin{lem}\label{gauss-dziala}
    Given an affine function \(h:\{0,1\}^n \rightarrow \{0,1\}^k\) of the form \(h(x)=Ax+b\) with \(\ker{h}\neq \emptyset\) we may construct in polynomial time
    an injective function \( g: \{0,1\}^d \rightarrow \{0,1\}^n \)
    of the form \( g(x)=Cx+p \) for some \( C \in \{0,1\}^{n\times d}, p \in \{0,1\}^n\), and \(d\), where \(d = \dim \ker{h}\), such that
    \( \Ima{g}=\ker{h} \).
    Moreover, we may choose \(C\) so that each of its columns has at most \(n-d+1\) ones.
    \end{lem}
    \begin{proof}
    We begin by obtaining an arbitrary affine parametrization of \(\ker{h}\). To this end fix some basis of \(\ker{A} \), arrange it as columns into the matrix \(C^{'}\) and any solution \(p\) to the equation \(Ax=-b\). All of this can be accomplished in polynomial time using Gaussian elimination \cite{gauss-elim}.
    Setting \(f(x)=C^{'}x+p\) gives us the desired parametrization.

    To reduce number of ones in the matrix \(C^{'} \), we can change basis of domain of \(f \) by an invertible matrix \(Q \in \{0,1\}^{d\times d}\).
    As function \(f\) is an injection, matrix \(C^{'}\) has \(d\) rows which are linearly independent and they form an invertible submatrix \(C^{''}\). By picking
    \(Q=(C^{''})^{-1}\) we assure that for each column of matrix \(C^{'}Q\) at most one row among those \(d\) picked previously has one which is contained in this column. After these steps we end up with a function \(g\) of the form \(g(x)=C^{'}Qx+p\), where each column of \(C=C^{'}Q\) has at most \(n-d+1\) non-zero entries.
    \end{proof}

    As we  parametrize the kernel of random affine map we want to make sure that dimension of this kernel is not too big, as the number of oracle queries and the number of non-oracle basic gates used by \(\texttt{SinglePoint}\) depends exponentially on the dimension of the searched space.

    \begin{lem}\label{dimensionality}
    If \(k<n-2\) then \(  \PP_{h \in_R H_{n,k}}(\dim{\ker{h}} \geq n-k+2)  \leq \frac{1}{16} \).
    \end{lem}
    \begin{proof}
    We prove the more general inequality \( \PP_{h \in_R H_{n,k}}(\dim{\ker{h}} \geq n-k+c)  \leq \frac{1}{2^{c^2}} \) for \(n>4\), \(k<n-2\) and any natural \(c \geq 2\).
    Set \(\delta = k-c \). If \(\delta < 0\) then the conclusion follows trivially. Otherwise
    the event \(\dim{\ker{h}} \geq n-k+c\)  is equivalent to \(  n-\delta = n-k+c  \leq \dim{\ker{h}} = \dim{\ker{A}} = n- \rank{A}\) so we conclude \(\rank{A}\leq \delta \), meaning that vector subspace spanned by rows of matrix \( A \) must have dimension at most \( \delta \).

    As vectors that span subspace of dimension at most \(\delta\) are contained in some \(\delta\)-dimensional subspace of \(\{0,1\}^n\) we can consider the probability of all \(k\) vectors being contained in a particular \(\delta\)-dimensional subspace. Then we apply union bound by multiplying this probability by the number of \(\delta\)-dimensional subspaces.
    The probability of choosing all \( k \) vectors from a single \( \delta \)-dimensional space is \( \paren*{\frac{1}{2^{n-\delta}}}^k \), as \( \delta\)-dimensional space contains \(2^{\delta}\) elements and we choose those vectors independently from each other.
    Let us recall that a number of \( \delta \)-dimensional subspaces of \(n\) dimensional space over \(\mathbb{F}_2\) equals \(  \prod_{i=0}^{\delta-1} \frac{2^n-2^i}{2^{\delta}-2^i} \), this can be found in \cite{subspaces}.

    So using the union bound the probability that \(k\) vectors span at most \( \delta \)-dimensional subspace is bounded from above by:
    \begin{multline*}
        \PP_{h_{A,b} \in H_{n,k}}\paren*{\dim{\ker{h}} \geq n-\delta}
        \\
        \leq \paren*{\frac{1}{2^{n-\delta}}}^k \prod_{i=0}^{(\delta-1)}
        \frac{2^n-2^i}{2^{\delta}-2^i}
        \\
        = \paren*{\frac{1}{2^{n-\delta}}}^{k-\delta}\paren*{\frac{1}{2^{n-\delta}}}^{\delta} \prod_{i=0}^{(\delta-1)} \frac{2^n-2^i}{2^{\delta}-2^i}
        \\
         = \paren*{\frac{1}{2^{n-\delta}}}^{k-\delta} \prod_{i=0}^{(\delta-1)} \frac{2^{\delta}-2^{i-n+\delta}}{2^{\delta}-2^i}
         \\
         \leq \paren*{\frac{1}{2^{n-\delta}}}^{k-\delta} \prod_{i=0}^{(\delta-1)} \frac{2^{\delta}}{2^{\delta}-2^i}
         \\
         = \paren*{\frac{1}{2^{n-\delta}}}^{k-\delta} \prod_{i=0}^{(\delta-1)} \frac{2^{\delta-i}}{2^{\delta-i}-1}
         \\
         \leq \paren*{\frac{1}{2^{n-\delta}}}^{k-\delta} \prod_{j=1}^{\infty} \frac{2^{j}}{2^{j}-1}
         \\
         = \paren*{\frac{1}{2^{n-\delta}}}^{k-\delta} \frac{1}{ \prod_{j=1}^{\infty} \paren*{1-2^{-j}} }
         \\
         \leq \paren*{\frac{1}{2^{n-\delta}}}^{k-\delta} \frac{1}{1-\frac{1}{2}-\frac{1}{4}}
         \\
         = 4\paren*{\frac{1}{2^{n-\delta}}}^{k-\delta} = 4 \paren*{\frac{1}{2^{n-k+c}}}^{c} = 4\paren*{\frac{1}{2^{n-k}}}^c \paren*{\frac{1}{2^c}}^c\text{,}
    \end{multline*}
    where the last inequality follows from Euler's pentagonal numbers theorem \cite{euler}.
    The final expression is less than \( \frac{1}{2^{c^{2}}} \) for \( c \geq 2 \) and \(n-k > 2\).
    \end{proof}

    Now we may describe the algorithm for solving unstructured search problem with known value \(k=1+\lceil \log_2{K} \rceil \) where \( K\) is the number of marked elements.

    \begin{algorithm}[H]
        \caption{Probabilistic algorithm for solving unstructured search with known \(k\)}
        \begin{algorithmic}[1]\label{multipoint}
            \captionsetup[algorithm]{name=MultiPoint}
            \Procedure{MultiPoint}{O, n, k}
                \If{\(k \geq n-2\)}
                    \State \( x \gets \) element from \( \{0,1\}^n \) selected uniformly at random
                    \If{ \( x \) is marked}
                        \State \Return x
                    \EndIf
                    \State \Return null
                \EndIf
                \State \(h \gets\) random affine transformation from \( H_{n,k}\)
                \State \(d \gets \dim{\ker{h}} \)
                \If{\( d \geq n-k+2\)}
                    \State \Return null
                \EndIf
                \State \(O_g \) is built as described in \Cref{oprim} using \(g\) from  \Cref{gauss-dziala}
                \State \Return \(\texttt{\texttt{SinglePoint}}(O_g, d)\)
            \EndProcedure
        \end{algorithmic}
    \end{algorithm}

    \multipoint*
    \begin{proof}
    To prove that \Cref{multipoint} finds a marked element with constant probability, let us see that if \(k \geq n-2\), then selecting a random element succeeds with probability at least \( \frac{1}{16} \).

    Otherwise from \Cref{valian-vazirani} with probability at least \(\frac{1}{8} \) we have that \(|K \cap \ker{h}| = 1 \). From \Cref{dimensionality}  with probability at least \( \frac{15}{16} \) we have that \( d < n-k+2 \).
    Combining those facts we obtain that with probability at least \( \frac{1}{16} \) oracle \(O_g\) marks exactly one element and the number of qubits of its input does not exceed \( n-k+1 \). So \Cref{multipoint} succeeds with probability at least \( \frac{1}{16} \), as from \Cref{thm:main} we know that \(\texttt{SinglePoint}\) solves the unstructured problem with one marked element with certainty.

    From \Cref{gauss-dziala} and \Cref{dekorator} we deduce that at most \(\calO(k)\)
     additional basic gates from circuit \(D_g\) depend on each qubit. So from \Cref{cor:average} we deduce that on average we use \(\calO(k)\) additional non-oracle basic gates per oracle query. There are \( \calO(2^{\frac{d}{2}}) = \calO(2^{\frac{n-k}{2}}) = \calO\paren*{\sqrt{\frac{N}{K}}}   \) oracle queries in \(\texttt{SinglePoint}\) procedure so  we need \( \calO(k2^{\frac{n-k}{2}}) =  \calO\paren*{\log{K}\sqrt{\frac{N}{K}}}\) non-oracle gates to implement it.
    \end{proof}

    \begin{prop}\label{prob-amp}
    For any \(p\in (0,1)\) by repeating \Cref{multipoint} \(\calO(\log{\frac{1}{1 - p}})\) number of times we assure that we find a marked element with probability at least \(p\).
    \end{prop}
    \begin{proof}
    We may deduce that from the fact that all runs of this algorithm are independent and each finishes successfully with constant, non-zero probability.
    \end{proof}

    Let us for any probability \(p<1\) define an algorithm \( \texttt{MultiPointAmplified}(O, n, k, p) \)
    which runs algorithm \(\texttt{MultiPoint}(O,n,k)\) minimal number of times to ensure probability of success higher than \(p\).

    \subsection{Unknown Number of Marked Elements}

    The technique presented next is similar to one used in \cite{inbook}, which finds element marked by an oracle \(O\) using on average \(\calO(2^{\frac{n-k}{2}})\) calls to oracle \(O\) and on average \(\calO(n2^{\frac{n-k}{2}})\) additional basic gates.
    We improve those result and propose the following algorithm that  finds marked element using
        \( \calO( k2^{\frac{n-k}{2}})\) non-oracle gates in expectation and makes \( \calO( 2^{\frac{n-k}{2}})\) queries to oracle \(O\) also in expectation.

    \begin{algorithm}[H]
        \caption{Probabilistic algorithm for solving unstructured search problem without an estimate of the number of marked elements}\label{alg:multipoint}
        \begin{algorithmic}[1]\label{multi-unknown}
            \Procedure{MultiPointUnknown}{O, n, p}
                \For{\(i \gets n+2\) to \(2\)}
                    \For{\( j \gets n+2\) to \(i\)}
                        \State \(x \gets \texttt{MultiPointAmplified}(O, n, j, p)\)
                        \If{\( x \) is marked}
                            \State \Return x
                        \EndIf
                    \EndFor
                \EndFor
            \EndProcedure
        \end{algorithmic}
    \end{algorithm}

    Before we analyze \Cref{multi-unknown} we note an observation:

    \begin{observation}\label{mnoznik}
    For natural numbers \(x\) and \( m \), and any real number \(r\), such that \(x \leq m\) and \(r > 1\), we have
\begin{multline*}
    \sum_{l=0}^{x}(m-l)r^{l/2} = (m-x)\sum_{l=0}^{x}r^{l/2} + \sum_{l=0}^{x}(x-l)r^{l/2}
    \\
    =
    (m-x)r^{x/2}\sum_{l=0}^{x}r^{(l-x)/2} + r^{x/2}\sum_{l=0}^{x}(x-l)r^{(l-x)/2}
    \\
    \leq (m-x)r^{x/2}\sum_{i=0}^{\infty}r^{-i/2} + r^{x/2}\sum_{i=0}^{\infty}ir^{-i/2}
    \\
    = C_1(m-x)r^{x/2}+C_2r^{x/2}
    \end{multline*}
    Where \(C_1,C_2\) are positive constants which depend only on \(r\).
    \end{observation}

    \begin{observation}\label{mnoznik-2}
    For natural numbers \(x\) and \( m \), and any real number \(r\), such that \(x \leq m\) and \(r < 1\), we have
    \[ \sum_{l=0}^{x}(m-l)r^{l/2}  \leq  Cm\] where \(C\) is a positive constant which depends only on \(r\).
    \end{observation}

    \begin{thm}
    For \( p \) satisfying \({2}(1-p)^2 < 1\) the \Cref{multi-unknown} finds marked element with probability at least \( 1-(1-p)^k \). Its expected number of oracle queries is \( \calO( \sqrt{\frac{N}{K}})\) and its expected number of non-oracle basic gates is \( \calO(\log{K} \sqrt{\frac{N}{K}})\), where \(N=2^n\) is the size of the search space and \(K\) is the number of elements marked by the oracle and \(k=1+\lceil \log_2{K} \rceil\).
    \end{thm}

    \begin{proof}
    In the complexity analysis we consider two phases of the \Cref{multi-unknown}. The first phase is when \( i > k\). During this phase we never run algorithm \(\texttt{MultiPointAmplified}(O, n, j,p)\) with \(j=k\), so let us assume that this algorithm never finds marked element in this phase.

    In the second phase i.e. for \( i < k\) in each inner loop we run the procedure \(\texttt{MultiPointAmplified}(O, n, j, p)\) with \(j=k\) once, so during this loop we find marked element with probability at least \(p\). So the probability that outer loop proceeds to the next iteration is at most \(1-p\).
    So overall bound on expected number of oracle queries of this algorithm is given below, we also note that all constants hidden under big \( \calO \) notation depend either only on \(p\) or are universal:

    \begin{multline*}
    \calO \paren*{\sum_{i=k+1}^{n+2} \sum_{j=i}^{n+2}2^{(n-j)/2} + \sum_{i=2}^{k}(1-p)^{k-i}\sum_{j=i}^{n+2}2^{(n-j)/2}}
    \\
    =
    \calO \paren*{\sum_{i=k}^{n} \sum_{j=i}^{n}2^{(n-j)/2} + \sum_{i=0}^{k}(1-p)^{k-i}\sum_{j=i}^{n}2^{(n-j)/2}}
    \\
    =
    \calO \paren*{\sum_{i=k}^{n} \sum_{l=0}^{n-i}2^{l/2} + \sum_{i=0}^{k}(1-p)^{k-i}\sum_{l=0}^{n-i}2^{l/2}}
    \\
    =
    \calO \paren*{\sum_{i=k}^{n} 2^{(n-i)/2} +\sum_{i=0}^{k}(1-p)^{k-i}2^{(n-i)/2}}
    \\
    =
      \calO\paren*{ \sum_{s=0}^{n-k} 2^{s/2} + 2^{(n-k)/2}\sum_{s=0}^{k}(2(1-p)^{2})^{s/2}}
      \\
      = \calO \paren*{2^{(n-k)/2}}
    \end{multline*}
    To estimate the second summand we use the fact that \(2(1-p)^2 < 1\).
    Using \Cref{mnoznik} and \Cref{mnoznik-2} we calculate the similar bound for the number of additional non-oracle basic gates, also take a note that hidden constants below depend only on \(p\) or are universal :

    \begin{widetext}
    \begin{multline*}
    \calO \paren*{\sum_{i=k+1}^{n+2} \sum_{j=i}^{n+2}j2^{(n-j)/2} + \sum_{i=2}^{k}(1-p)^{k-i}\sum_{j=i}^{n+2}j2^{(n-j)/2}}
    =
    \calO \paren*{\sum_{i=k}^{n} \sum_{j=i}^{n}j2^{(n-j)/2} + \sum_{i=0}^{k}(1-p)^{k-i}\sum_{j=i}^{n}j2^{(n-j)/2}}
    \\
    =
    \calO \paren*{\sum_{i=k}^{n} \sum_{l=0}^{n-i}(n-l)2^{l/2} + \sum_{i=0}^{k}(1-p)^{k-i}\sum_{l=0}^{n-i}(n-l)2^{l/2}}
    =
    \calO \paren*{\sum_{i=k}^{n}i2^{(n-i)/2} +\sum_{i=0}^{k}i(1-p)^{k-i}2^{(n-i)/2}}
    \\
    =
      \calO\paren*{ \sum_{s=0}^{n-k}(n-s)2^{s/2} + 2^{(n-k)/2}\sum_{s=0}^{k}(k-s)(2(1-p)^{2})^{s/2}}
      = \calO \paren*{k2^{(n-k)/2}}
    \end{multline*}
    \end{widetext}

    So the complexity of the algorithm does not change even if the number of elements is not known beforehand.
    To calculate the probability of successfully finding the marked element let us see that the outer loop runs less than $n+1$ times only when the marked element was found. From the above considerations we know that this probability is bounded from below by \(1-(1-p)^k\).
    \end{proof}

    \begin{acknowledgements}
        We would like to express our deep gratitude to our friends at Beit, in particular to Wojciech Burkot, for their insights and criticism.
        However, mere language would not suffice for this endeavour, so we will refrain from doing so.
    \end{acknowledgements}
    \bibliographystyle{plain}
    \addcontentsline{toc}{section}{References}
    \bibliography{references}

\begin{thebibliography}{10}

\bibitem{arora}
Sanjeev Arora and Boaz Barak.
\newblock {\em Computational Complexity: A Modern Approach}.
\newblock Cambridge University Press, USA, 1st edition, 2009.

\bibitem{wolfsearch}
Srinivasan Arunachalam and Ronald de~Wolf.
\newblock Optimizing the number of gates in quantum search.
\newblock {\em arXiv preprint arXiv:1512.07550}, 2015.

\bibitem{inbook}
Michel Boyer, Gilles Brassard, Peter Høyer, and Alain Tapp.
\newblock Tight bounds on quantum searching.
\newblock {\em Fortschritte der Physik}, 46(4-5):493–505, Jun 1998.

\bibitem{ampamp}
Gilles Brassard, Peter Høyer, Michele Mosca, and Alain Tapp.
\newblock Quantum amplitude amplification and estimation.
\newblock {\em Contemporary Mathematics}, 305:53--74, 2002.

\bibitem{Brassard_1998}
Gilles Brassard, Peter Høyer, and Alain Tapp.
\newblock Quantum cryptanalysis of hash and claw-free functions.
\newblock {\em Lecture Notes in Computer Science}, page 163–169, 1998.

\bibitem{patent}
Wojciech Burkot, Jan Tułowiecki, Vladyslav Hlembotskyi, and Witold Jarnicki.
\newblock Quantum circuit and methods for use therewith.
\newblock Mar 2020.
\newblock US patent application No. 62990122.

\bibitem{uniquesat}
Chris Calabro, Russell Impagliazzo, Valentine Kabanets, and Ramamohan Paturi.
\newblock The complexity of unique k-sat: An isolation lemma for k-cnfs.
\newblock {\em Journal of Computer and System Sciences}, 74(3):386--393, 2008.

\bibitem{carter}
J.~Lawrence Carter and Mark~N. Wegman.
\newblock Universal classes of hash functions.
\newblock {\em Journal of Computer and System Sciences}, 18(2):143 -- 154,
  1979.

\bibitem{durr1996quantum}
Christoph Durr and Peter Høyer.
\newblock A quantum algorithm for finding the minimum, 1996.

\bibitem{D_rr_2006}
Christoph Dürr, Mark Heiligman, Peter Høyer, and Mehdi Mhalla.
\newblock Quantum query complexity of some graph problems.
\newblock {\em SIAM Journal on Computing}, 35(6):1310–1328, Jan 2006.

\bibitem{euler}
Leonhard Euler.
\newblock Evolutio producti infiniti \( (1 - x)(1 - xx)(1 - x^3)(1 - x^4)(1 -
  x^5)(1 - x^6) \) etc. in seriem simplicem.
\newblock {\em Acta Academiae Scientarum Imperialis Petropolitinae}, pages
  47--44, 1783.

\bibitem{subspaces}
Jay Goldman and Gian-Carlo Rota.
\newblock On the foundations of combinatorial theory iv finite vector spaces
  and eulerian generating functions.
\newblock {\em Studies in Applied Mathematics}, 49(3):239--258, 1970.

\bibitem{grover96}
Lov~K. Grover.
\newblock A fast quantum mechanical algorithm for database search.
\newblock In {\em Proceedings of the twenty-eighth annual ACM symposium on
  Theory of computing}, pages 212--219, 1996.

\bibitem{groverfast}
Lov~K. Grover.
\newblock Trade-offs in the quantum search algorithm.
\newblock {\em Physical Review A}, 66(5):052314, 2002.

\bibitem{eksperimental}
Jan Gwinner, Marcin Briański, Wojciech Burkot, Łukasz Czerwiński, and
  Vladyslav Hlembotskyi.
\newblock Benchmarking 16-element quantum search algorithms on ibm quantum
  processors, 2020.

\bibitem{ibm}
A.~{Mandviwalla}, K.~{Ohshiro}, and B.~{Ji}.
\newblock Implementing grover’s algorithm on the ibm quantum computers.
\newblock In {\em 2018 IEEE International Conference on Big Data (Big Data)},
  pages 2531--2537, 2018.

\bibitem{toeplitz}
Yishay Mansour, Noam Nisan, and Prasoon Tiwari.
\newblock The computational complexity of universal hashing.
\newblock {\em Theoretical Computer Science}, 107(1):121 -- 133, 1993.

\bibitem{gauss-elim}
M~Thamban Nair and Arindama Singh.
\newblock Elementary operations.
\newblock In {\em Linear Algebra}, pages 107--161. Springer, 2018.

\bibitem{nielsen}
Michael~A. Nielsen and Isaac Chuang.
\newblock Quantum computation and quantum information.
\newblock {\em American Journal of Physics}, 70(5):558--559, 2002.

\bibitem{valiant}
L.~G. Valiant and V.~V. Vazirani.
\newblock Np is as easy as detecting unique solutions.
\newblock In {\em Proceedings of the Seventeenth Annual ACM Symposium on Theory
  of Computing}, STOC ’85, page 458–463, New York, NY, USA, 1985.
  Association for Computing Machinery.

\bibitem{zalka}
Christof Zalka.
\newblock Grover’s quantum searching algorithm is optimal.
\newblock {\em Physical Review A}, 60(4):2746–2751, Oct 1999.

\bibitem{zhang}
Kun Zhang and Vladimir~E. Korepin.
\newblock Depth optimization of quantum search algorithms beyond grover’s
  algorithm.
\newblock {\em Physical Review A}, 101(3), Mar 2020.

\end{thebibliography}
    \newpage
    \appendix
\makeatletter\onecolumngrid@push\makeatother
\section{Analysis of the Tree circuit}
\label{appendix:a}

Due to limited nature of existing hardware, for small search spaces the circuits \(W_m\) are outperformed by a similar family of circuits, which we denote by \(D_m\). The circuits were experimentally evaluated on current generation of superconducting quantum computers. The results are presented in \cite{eksperimental}. For the sake of completeness we prove an analogue of \cref{thm:main} that utilises the \(D_m\) family of circuits.

\begin{ddef}\label{def:dszewo}
    Let $\overline{k} = (k_1, \dots, k_m)$ be a sequence of positive integers and let $n := \sum_{j=1}^m k_j$. Given a quantum oracle $O$, for $j \in \set{0, \dots, m}$ we define the circuit $D_j$ recursively as follows:
    \[
        D_j =
        \begin{cases}
            \id_{n} & \quad \text{ if } \ j = 0 \\
            D_{j - 1} \cdot \paren*{\id_{k_1 + \dots + k_{j - 1}} \otimes G_{k_j} \otimes \id_{k_{j + 1} + \dots + k_m}} \cdot O \cdot D_{j - 1} & \quad \text{ if } \ j \neq 0.
        \end{cases}
    \]
\end{ddef}

\begin{lem}\label{lem:symmetry}
    Let \(m \in \mathbb{N}_{+}\) and \(\overline{k} \in \mathbb{N}_{+}^m\) be fixed, and let \(n = \sum_{j = 1}^{m} k_j\). Assume we are given a phase oracle \(O\) that operates on \(n\) qubits and marks a single vector of the standard computational basis, which we then use in the circuits \(D_j\).
    Then for any \(j \in \set{0, \dots, m}\) we have
    \[
        D_j O D_j = O \text{.}
    \]
\end{lem}
\begin{proof}
    We proceed by induction on \(j\). For \(j = 0\) the claim is trivial. For \(j > 0\) we expand \(D_j\) according to \cref{def:dszewo} as follows
    \begin{align}
        D_j O D_j   &= D_{j - 1} \paren*{\id_{s - k_j} \otimes G_{k_j} \otimes \id_{n - s}} O D_{j - 1} O  D_{j - 1} \paren*{\id_{s - k_j} \otimes G_{k_j} \otimes \id_{n - s}} O D_{j - 1}  \nonumber \\
                    &= D_{j - 1} \paren*{\id_{s - k_j} \otimes G_{k_j} \otimes \id_{n - s}} O O \paren*{\id_{s - k_j} \otimes G_{k_j} \otimes \id_{n - s}} O D_{j - 1} \label{eqn:us1}\\
                    &= D_{j - 1} \paren*{\id_{s - k_j} \otimes G_{k_j} \otimes \id_{n - s}} \paren*{\id_{s - k_j} \otimes G_{k_j} \otimes \id_{n - s}} O D_{j - 1} \nonumber \\
                    &= D_{j - 1}  O D_{j - 1} \nonumber \\
                    &= O \label{eqn:us2}
    \end{align}
    where in \cref{eqn:us1,eqn:us2} we used the inductive hypothesis.
\end{proof}

\begin{lem}\label{lem:recurrenceTree}
    Let \(m \in \mathbb{N}_{+}\) and \(\overline{k} \in \mathbb{N}_{+}^m\) be fixed, and let \(n = \sum_{j = 1}^{m} k_j\). Assume we are given a phase oracle \(O\) that operates on \(n\) qubits and marks a single vector of the standard computational basis denoted \(\target\).
    Define the numbers 
    \[
        \beta_j = \bra{\target} \paren*{D_j \ket{u_1^j} \ket{\target_{j + 1}^{m}}}
    \]
    for \(j \in \set{0, \dots, m}\). Then \(\beta_j\) satisfy the recurrence
    \[
        \beta_j = \begin{cases}
            1, & \ \text{ if }\ j = 0\\
            \frac{1}{2^{ (k_1 + \dots + k_j)/ 2}} \paren*{1 - \frac{2}{2^{k_j}}} + \frac{1}{2^{k_j / 2}} \paren*{2 - \frac{2}{2^{k_j / 2}}} \beta_{j - 1}, & \ \text{ if }\ j > 0 \text{.}
        \end{cases}
    \]
\end{lem}
\begin{proof}
    By defintion of \(D_j\) we have \(\beta_0 = 1\) giving the base case.
    For \(j > 0\), we proceed to compute \(\beta_j\) by expanding the circuit \(D_j\) according to the recursive definition. We split the computation into stages as follows
    \begin{align*}
        \ket{w_1} &= D_{j-1} \paren*{\ket{u_1^j} \ket{\target_{j + 1}^{m}}}\\
        \ket{w_2} &= O \ket{w_1}\\
        \ket{w_3} &=  \paren*{\id_{s - k_j} \otimes G_{k_j} \otimes \id_{n - s}}\ket{w_2}\\
        \ket{w_4} &= D_{j - 1} \ket{w_3}
    \end{align*}
    where \(s = k_1 + \dots + k_j\).
    
    \begin{align}
        \ket{w_1}   &= D_{j - 1} \paren*{\frac{1}{2^{k_j / 2}}\ket{u_1^{j - 1}} \ket{\target_j^m} + \ket{u_1^{j - 1}} \ket{\overline{\target}_j} \ket{\target_{j + 1}^m}} \nonumber \\
                    &= \frac{1}{2^{k_j / 2}} D_{j - 1}\ket{u_1^{j - 1}} \ket{\target_j^m} + \ket{u_1^{j - 1}} \ket{\overline{\target}_j} \ket{\target_{j + 1}^m} \nonumber \\
        \ket{w_2}   &= O \paren*{\frac{1}{2^{k_j / 2}} D_{j - 1}\ket{u_1^{j - 1}} \ket{\target_j^m} + \ket{u_1^{j - 1}} \ket{\overline{\target}_j} \ket{\target_{j + 1}^m}} \nonumber \\
                    &= \frac{1}{2^{k_j / 2}} O D_{j - 1}\ket{u_1^{j - 1}} \ket{\target_j^m} + \ket{u_1^{j - 1}} \ket{\overline{\target}_j} \ket{\target_{j + 1}^m} \nonumber \\
                    &= \frac{1}{2^{k_j / 2}} \ket{\eta} \ket{\target_j^m} + \ket{u_1^{j - 1}} \ket{\overline{\target}_j} \ket{\target_{j + 1}^m} \nonumber \\
        \intertext{
        Where \(\ket{\eta}\) is some state in \(\paren*{\mathbb{C}^2}^{\otimes \paren*{k_1 + \dots + k_{j - 1}}}\). We can write so, as all diffusers within \(D_{j - 1}\) operate only on the prefix consisting of first \(k_1 + \dots + k_{j - 1}\) qubits, while \(O\) only changes relative phases.
        }
        \ket{w_3}   &= \id_{s - k_j} \otimes G_{k_j} \otimes \id_{n - s} \paren*{\frac{1}{2^{k_j / 2}} \ket{\eta} \ket{\target_j^m} + \ket{u_1^{j - 1}} \ket{\overline{\target}_j} \ket{\target_{j + 1}^m}} \nonumber \\
                    &= \frac{1}{2^{k_j / 2}} \ket{\eta} \paren*{G_{k_j} \ket{\target_j}} \ket{\target_{j + 1}^m} + \ket{u_1^{j - 1}} \paren*{G_{k_j}\ket{\overline{\target}_j}} \ket{\target_{j + 1}^m} \nonumber \\
                    &= \frac{1}{2^{k_j / 2}} \ket{\eta} \paren*{\frac{2}{2^{k_j / 2}} \ket{u_j} - \ket{\target_j}} \ket{\target_{j + 1}^m} + \ket{u_1^{j - 1}} \paren*{\paren*{1 - \frac{2}{2^{k_j}}} \ket{u_j} + \frac{1}{2^{k_j / 2}} \ket{\target_j}} \ket{\target_{j + 1}^m} \nonumber \\
                    &= \frac{1}{2^{k_j / 2}} \paren*{ \paren*{\frac{2}{2^{k_j}} - 1} \ket{\eta} + \paren*{2 - \frac{2}{2^{k_j}}} \ket{u_1^{j - 1}}} \ket{\target_j^m} + \paren*{\frac{2}{2^{k_j}} \ket{\eta} + \paren*{1 - \frac{2}{2^{k_j}}} \ket{u_1^{j - 1}}} \ket{\overline{\target}_j} \ket{\target_{j + 1}^{m}} \nonumber \\
        \ket{w_4}   &= D_{j - 1} \ket{w_3} \nonumber \\
                    &= \frac{1}{2^{k_j / 2}} \paren*{ \paren*{\frac{2}{2^{k_j}} - 1} O\ket{u_{1}^{j - 1}}\ket{\target_j^m} + \paren*{2 - \frac{2}{2^{k_j}}} D_{j - 1}\ket{u_1^{j - 1}} \ket{\target_{j}^m}} \nonumber\\ & \quad \quad + \paren*{\frac{2}{2^{k_j}} \ket{\eta} + \paren*{1 - \frac{2}{2^{k_j}}} \ket{u_1^{j - 1}}} \ket{\overline{\target}_j} \ket{\target_{j + 1}^{m}} \nonumber
    \end{align}
    Note that we used \cref{lem:symmetry} when applying \(D_{j - 1}\) in the first summand.
    Now we can plug \(\ket{w_4}\) into the expression defining \(\beta_j\). Observe that the second summand in final expression is orthogonal to \(\ket{\target}\), thus can be safely discarded. We obtain
    \begin{align}
        \beta_j     &= \bra{\target} \paren*{\frac{1}{2^{k_j / 2}} \paren*{ \paren*{\frac{2}{2^{k_j}} - 1} O\ket{u_{1}^{j - 1}}\ket{\target_j^m} + \paren*{2 - \frac{2}{2^{k_j}}} D_{j - 1}\ket{u_1^{j - 1}} \ket{\target_{j}^m}}} \nonumber \\
                    &= \frac{1}{2^{k_j / 2}} \paren*{ \paren*{1 - \frac{2}{2^{k_j}}} \frac{1}{2^{\paren*{s - k_j} / 2}} + \paren*{2 - \frac{2}{2^{k_j}}} \beta_{j - 1}} \nonumber \\
                    &= \frac{1}{2^{ \paren*{k_1 + \dots + k_j}/ 2}} \paren*{1 - \frac{2}{2^{k_j}}} + \frac{1}{2^{k_j / 2}} \paren*{2 - \frac{2}{2^{k_j / 2}}} \beta_{j - 1} \nonumber
    \end{align}
    concluding the proof.
\end{proof}
      \begin{thm}\label{thm:mainv2}
            Fix any \(\varepsilon > 0\), and any \(N \in \mathbb{N}\) of the form $N = 2^n$. Suppose we are given a quantum oracle \(O\) operating on \(n\) qubits that marks exactly one element. Then there exists a quantum circuit \(\calA\) which uses the oracle \(O\) at most
            \(\paren*{1 + \varepsilon} \frac{\pi}{4} \sqrt{N}\) times
            and uses at most
            \(\calO \paren*{\log\paren*{\sfrac{1}{\varepsilon}} \sqrt{N}}\)
            non-oracle basic gates, which finds the element marked by \(O\) with certainty.
        \end{thm} 
\begin{proof}
    We first begin by choosing a particular sequence of sizes for diffusers in the circuit \(D_m\) -- namely \(k_j = (x + 1) \cdot j\) where \(x \in \mathbb{N}_{+}\) is some parameter, and let us assume that the number of qubits we work with is precisely \((x + 1) + (x + 1) \cdot 2 + \dots + (x + 1) \cdot m = (x + 1) m (m + 1) / 2\).
    From \cref{lem:recurrenceTree}, we get that the amplitude the circuit \(D_m\) in the marked state is given by the following recurrence
    \[
        \beta_j = \begin{cases}
            1, & \ \text{ if }\ j = 0\\
            \frac{1}{2^{ (x + 1) j (j + 1) / 4}} \paren*{1 - \frac{2}{2^{(x + 1)j}}} + \frac{1}{2^{(x + 1) j / 2}} \paren*{2 - \frac{2}{2^{(x + 1) j / 2}}} \beta_{j - 1}, & \ \text{ if }\ j > 0 \text{.}
        \end{cases}
    \]
    To simplify the analysis of this recurrence, let us begin by substituting \(\gamma_j = \beta_j \cdot 2^{{(x + 1) j (j + 1)}/{4}} \cdot 2^{-j}\), which yields
    \[
        \gamma_j = \begin{cases}
            1, & \ \text{ if }\ j = 0\\
            \paren*{1 - 2 \cdot 2^{-(x + 1) j}} \cdot 2^{-j} + \paren*{1 - 2 ^{- (x + 1) j}} \gamma_{j - 1}, & \ \text{ if }\ j > 0 \text{.}
        \end{cases}
    \]
    We easily obtain the following inequality for \(j > 0\)
    \[
        \gamma_j \geq \paren*{1 - 2^{-xj}} \paren*{2^{-j} + \gamma_{j - 1}} \text{.}
    \]
    
    We may thus set
    \[
        \delta_j = \begin{cases}
            1, & \ \text{ if }\ j = 0\\
            \paren*{1 - 2^{-xj}} \paren*{2^{-j} + \delta_{j - 1}}, & \ \text{ if } \ j > 0
        \end{cases}
    \]
    and we easily obtain the inequality \(\gamma_j \geq \delta_j\). We can express the solution to this recurrence as a sum
    \[
        \delta_m = \sum_{j = 0}^{m - 1} 2^{j - m} \cdot \prod_{k = m - j}^{m} \paren*{1 - q^k} + \prod_{k = 1}^{m} \paren*{1 - q^k}
    \]
    where \(q = 2^{-x}\).
    
    For \(a \in \mathbb{N} \cup \set{\infty}\), let
    \[
        \calP(a) = \prod_{k = 1}^{a} \paren*{1 - q^k} \text{.}
    \]
    In terms of \(\calP(a)\), we can lower bound \(\delta_m\), as each term in our product is strictly less than 1, thus
    \[\delta_m \geq \sum_{j = 0}^{m - 1} 2^{-m + j} \calP(m) + \calP(m) = \sum_{j = 0}^{m} 2^{-j} \calP(m) = (2 - 2^{-m}) \calP(m) \geq (2 - 2^{-m}) \calP(\infty). \]
    We now need a lower bound on \(\calP(\infty)\), which we can obtain via Euler's Pentagonal Number Theorem \cite{euler}, which states that
    \[\calP(\infty) = \prod_{k=1}^{\infty} \left(1 - q^k \right)
    = 1 + \sum_{k = 1}^{\infty} (-1)^k \left(q^{{(3k -1) k}/{2}} + q^{{(3k + 1)k}/{2}} \right)\]
    from which, one can easily derive the inequality
    \[\calP(\infty) \geq 1 - q - q^2 \text{.}\]
    Combining these inequalities we get
    \begin{equation}
        \beta_m \geq (2 - 2^{-m}) \left(1 - 2^{-x} - 2^{-2x}\right) \cdot 2^{m} \cdot 2^{-{n}/{2}} \text{.}
        \label[ineq]{ineq:alp}
    \end{equation}
    Using the same reasoning as in the proof of \cref{thm:main}, the \cref{ineq:alp} allows us to bound the number of iterations of amplitude amplification by
    \[\frac{\pi}{4} \cdot \frac{1}{2 - 2^{-m}} \cdot \frac{1}{1 - q - q^2} \cdot 2^{-m} \cdot 2^{n / 2}.\]
    Each \(D_m\) has exactly \(2^m - 1\) oracle calls, so one iteration has \(2^{m + 1} - 1\) oracle calls (tree, its conjugate and 1 extra call). Thus the number of oracle calls is at most
    \[\frac{\pi}{4} \cdot \frac{1}{2 - 2^{-m}} \cdot \frac{1}{1 - q - q^2} \cdot 2^{-m} \cdot 2^{{n}/{2}} \cdot \left(2^{m + 1} - 1\right) = \frac{\pi}{4}\cdot \frac{1}{1 - q - q^2}\cdot 2^{{n}/{2}}\]
    so we are only a factor of \(\frac{1}{1 - 2^{-x} - 2^{-2x}}\) away from optimal number of oracle calls.
    
    $D_m$ can be implemented with \(\calO(\sum_{k = 1}^{m} ky \cdot 2^{m - k})\) gates. So we get at most
    \[\calO\left( \left(\sum_{k = 1}^{m}kx \cdot 2^{m - k}\right) \cdot 2^{-m} \cdot 2^{{n}/{2}} \right) = \calO \left( \left(\sum_{k = 1}^{m}kx 2^{-k}\right) 2^{{n}/{2}} \right)\]
    non-oracle gates used by our algorithm. We use the following simple observation
    \[\sum_{k = 1}^{m} kx 2^{-k} \leq \sum_{k=1}^{\infty}kx 2^{-k} = 2x\]
    to get that the total number of nonoracle gates used by our algorithm is bounded by \(\calO \left( x \cdot 2^{{n}/{2}} \right)\).
    Thus, for any \(\varepsilon > 0\) that is sufficiently small, we obtain an algorithm that makes at most
    \[(1 + \varepsilon) \frac{\pi}{4} \cdot 2^{{n}/{2}}\]
    oracle calls, and uses at most
    \[\calO \left( \log \left( \varepsilon^{-1} \right) 2^{{n}/{2}} \right)\]
    non-oracle gates by setting \(x \in \Theta\paren*{\log\paren*{\varepsilon^{-1}}}\).
\end{proof}

\end{document}